\documentclass{llncs}
\usepackage[a4paper, margin=25.0mm]{geometry}
\usepackage{graphicx}

\usepackage[english]{babel}
\usepackage{amsfonts}
\usepackage{inputenc}
\usepackage{color}
\usepackage{algpseudocode}
\usepackage{booktabs,multicol,multirow}
\usepackage{fancyhdr}
\usepackage{epic}
\usepackage{xspace}
\usepackage{amssymb,amsmath,verbatim}
\usepackage{multirow}
\usepackage{multicol}
\usepackage{adjustbox}
\usepackage{threeparttable}
\usepackage{float}

\usepackage{amsthm}
\usepackage{chngcntr}

\usepackage[linesnumbered,ruled,vlined]{algorithm2e}
\SetKwInput{KwInput}{Input}                
\SetKwInput{KwOutput}{Output}              

\newtheorem{thm}{Theorem}

\newtheorem{lem}[thm]{Lemma}
\newtheorem{rem}[thm]{Remark}

\newtheorem{defn}[thm]{Definition}

\usepackage{url}

\def\ZZ{\mathbb{Z}}
\def\QQ{\mathbb{Q}}
\def\RR{\mathbb{R}}

\def\cal{\mathcal}
\def\bf{\mathbf}

\def\TrapGen{\mathsf{TrapGen}}
\def\PK{\mathsf{PK}}
\def\SK{\mathsf{SK}}
\def\CT{\mathsf{CT}}
\def\SampleLeft{\mathsf{SampleLeft}}
\def\SampleRight{\mathsf{SampleRight}}
\def\SamplePre{\mathsf{SamplePre}}
\def\SampleBasisLeft{\mathsf{SampleBasisLeft}}

\def\SampleP{\mathsf{SampleP}}

\def\td{\mathsf{td}}

\def\Setup{\mathsf{Setup}}

\def\Enc{\mathsf{Enc}}
\def\Dec{\mathsf{Dec}}
\def\Td{\mathsf{Td}}
\def\Test{\mathsf{Test}}
\def\Pr{\mathrm{Pr}}
\def\Adv{\mathsf{Adv}}
\def\OW{\textsf{OW-CCA2}}
\def\IND{\textsf{IND-CCA2}}
\def\OWa{\textsf{OW-CCA}}
\def\INDa{\textsf{IND-CCA}}

\def\a{\alpha}
\def\u{\bf{u}}

\def\t{\theta}
\def\e{\bf{e}}

\def\x{\bf{x}}

\def\L{\Lambda}
\def\Lp{\Lambda^{\perp}}
\def\b{\bf{b}}
\def\s{\bf{s}}
\def\c{\bf{c}}

\def\RLWE{\textbf{Ring-LWE}_{n,q,D_{R,\sigma}}}

\begin{document}
	\title{Lattice-based public key encryption with equality test supporting flexible authorization in standard model}
	\titlerunning{}
	\authorrunning{}
	\author{Dung Hoang Duong\inst{1}\and Kazuhide Fukushima\inst{2}\and Shinsaku Kiyomoto\inst{2}\and\\ Partha Sarathi Roy\inst{1} \and Arnaud Sipasseuth\inst{2} \and Willy Susilo\inst{1}}
	\institute{Institute of Cybersecurity and Cryptology\\
	School of Computing and Information Technology, University of Wollongong\\
		Northfields Avenue, Wollongong NSW 2522, Australia\\
		\email{\{hduong,partha,wsusilo\}@uow.edu.au} 
		\and
		Information Security Laboratory, KDDI Research, Inc.\\
		2-1-15 Ohara, Fujimino-shi, Saitama, 356-8502, Japan\\
		\email{\{ka-fukushima,kiyomoto,ar-sipasseuth\}@kddi-research.jp}
	}

	\maketitle              
	
	\begin{abstract}
Public key encryption with equality test (PKEET) supports to check whether two ciphertexts encrypted under different public keys contain the same message or not. PKEET has many interesting applications such as keyword search on encrypted data, encrypted data partitioning for efficient encrypted data management, personal health record systems, spam filtering in encrypted email systems and so on. However, the PKEET scheme lacks an authorization mechanism for a user to control the comparison of its ciphertexts with others. In 2015, Ma et al. introduce the notion of PKEET with flexible authorization (PKEET-FA) which strengthens privacy protection. Since 2015, there are several follow-up works on PKEET-FA. But, all are secure in the random-oracle model. Moreover, all are vulnerable to quantum attacks. In this paper, we provide three constructions of quantum-safe PKEET-FA secure in the standard model. Proposed constructions are secure based on the hardness assumptions of integer lattices and ideal lattices. Finally, we implement the PKEET-FA scheme over ideal lattices.
		
	\end{abstract}

\section{Introduction}
Public key encryption with equality test (PKEET), which supports the equality test of underlying messages of two ciphertexts with the help of trapdoor, was first introduced by Yang et al.~\cite{Yang10}. The property of equality test is of use for various practical applications, such as keyword search on encrypted data, encrypted data partitioning for efficient encrypted data management, personal health record systems, spam filtering in encrypted email systems and so on. But, to protect the privacy of the data owner, it is required to have an authorization mechanism. In this direction, the first successful effort was by Tang \cite{Tang11}, where a fine-grained authorization policy enforcement mechanism was integrated into PKEET. This work was further extended by using two collaborating proxies to perform the equality tests \cite{Tang12}. Moreover, an all-or-nothing PKEET was presented by Tang \cite{Tang12b}. But, all the above schemes are one-way against chosen-ciphertext attack (OW-CCA). However, for some special scenarios, such as database applications, OW-CCA security may not be strong enough. Motivated by this, the security models of PKEET were revisited by Lu et al. \cite{LZL12}, who proposed several new and stronger security definitions. A new notion, called public key encryption with delegated equality test (PKE-DET), was introduced by Ma et al. \cite{Ma15}, where the delegated party is only allowed to deal with the work in a multi-user setting. Ma et al. \cite{Ma15b} proposed a new primitive, called public key encryption with equality test supporting flexible authorization (PKEET-FA), to strengthen privacy protection. Lin et al. \cite{LQZ16}, remove the pairing computation from the construction of Ma et al. \cite{Ma15b}. However, they are all proven to be secure in the random oracle model based on the number-theoretic hardness assumption. It is a need of the age to construct secure primitive based on post-quantum secure hardness assumption. Therefore it is necessary to construct such a scheme in the standard model based on post-quantum secure hardness assumption. Recently, Lee et al. \cite{Lee2016} proposed a generic construction of PKEET in the standard model from which it is possible to mount a lattice-based construction. Duong et al. \cite{Duong19} proposed a concrete construction of PKEET secure in the standard model based on the hardness of LWE problem. Both of these two works do not consider the authorization mechanism.

\noindent\textbf{Our contribution:}
In this paper, our contribution is fourfold:
\begin{itemize}
	\item According to the best of our knowledge, we propose the first concrete construction of a PKEET-FA scheme based on the hardness assumption of integer lattices.
	\item We describe a lattice-based instantiation of PKEET-FA from Lee et al.'s \cite{Lee2016} generic construction.
	\item To have better efficiency in respect of the size of parameters, we propose a construction of PKEET-FA scheme based on the hardness assumption of ideal lattices.
	\item Finally, we implement the PKEET-FA scheme over ideal lattices.
\end{itemize}
We first employ the multi-bit full IBE by Agrawal et al.~\cite{ABB10-EuroCrypt} and then directly transform it into a PKEET scheme. Finally, we have devoleped the authorization algorithms to construct PKEET-FA over integer lattices. In our scheme, a ciphertext is of the form $\CT=(\CT_1,\CT_2,\CT_3,\CT_4)$ where $(\CT_1,\CT_3)$ is the encryption of the message $\bf{m}$, as in the original IBE scheme, and $(\CT_2,\CT_4)$ is the encryption of $H(\bf{m})$ in which $H$ is a hash function. In order to utilize the IBE scheme, we employ a second hash function $H'$ and create the \textit{identity} $H'(\CT_1,\CT_2)$ before computing $\CT_3$ and $\CT_4$. 
To instantiate PKEET-FA from \cite{Lee2016}, we have employed the HIBE from \cite{ABB10-EuroCrypt}. Finally, we have used the ideal version of IBE of \cite{ABB10-EuroCrypt} from \cite{BertFRS18-implement} to construct PKEET-FA over ideal lattices.
As compared to the previous constructions, the proposed constructions are computationally efficient due to the absence of exponentiation and evaluation of pairing; see Table \ref{tab1}. Moreover, all the previous constructions are secure in the random oracle model, whereas, all the proposed constructions are secure in the standard model.

{\renewcommand{\arraystretch}{1.3}
	\setlength{\tabcolsep}{1.5pt}
	\begin{table*}[ht]
		\begin{center}
			\caption{Efficency Comparison with Existing PKEET.}
			\begin{tabular}{c c c c c c c}
				\toprule
				Scheme & $\cal{C}_{\Enc}$ & $\cal{C}_{\Dec}$ & $\cal{C}_{\Td}$ & $\cal{C}_{\Test}$ & Std./RO & Asmp\\
				\midrule
				\cite{Yang10} & 3 Exp &  3 Exp & - & 2 Pairing & RO & CDH\\
				\cite{Tang11,Tang12} & 4 Exp &  2 Exp & 3 Exp & 4 Pairing & RO & CDH + DDH\\
				\cite{Tang12b} & 5 Exp & 2 Exp & 0 & 4 Exp & RO & CDH\\
				\cite{Ma15}-Type 1 & 6 Exp &  5 Exp & 0 & 2 Pairing+ 2 Exp & RO & CDH \\
				\cite{Ma15}-Type 2 & 6 Exp &  5 Exp & 2 Exp & 2 Pairing+ 2 Exp & RO & CDH \\
				\cite{Ma15}-Type 3 & 6 Exp &  5 Exp & 1 Exp & 2 Pairing+ 2 Exp & RO & CDH \\
				\cite{LQZ16}-Type 1 & 4 Exp & 3 Exp & 0 & 2 Exp & RO & DDH\\
				\cite{LQZ16}-Type 2 & 4 Exp & 3 Exp & 1 Exp & 0 & RO & DDH\\
				\cite{LQZ16}-Type 3 & 4 Exp & 3 Exp & 1 Exp & 1 Exp & RO & DDH\\
				Proposed- Type 1,2, 3 & 0 &  0 & 0 & 0 & Std. & LWE\\
                (Over Integer Lattices) &  &  &  &  &  & \\	
                Proposed- Type 1 & 0 &  0 & 0 & 0 & Std. & LWE\\
                (Instantiation of  &  &  &  &  &  & \\
                Lee et al.'s \cite{Lee2016}construction) &  &  &  &  &  & \\
                Proposed- Type 1,2, 3 & 0 &  0 & 0 & 0 & Std. & RLWE\\
                (Over Ideal Lattices) &  &  &  &  &  & \\			
				\toprule
				
				\noalign{\smallskip}
				
				\multicolumn{7}{l}{${}^{\phantom{**}}$ \parbox[t]{0.99\textwidth}{$\cal{C}_{\Enc}$, $\cal{C}_{\Dec}$, $\cal{C}_{\Td}$, $\cal{C}_{\Test}$: the computation complexity of algorithms for encryption, decryption, $\Td$ and $\Test$; Asmp: assumption; Exp: exponentiation; pairing: pairing evaluation; Std./RO: security in standard model or random oracle model; CDH: computational Diffie-Hellman assumption; DDH: decesional Diffie-Hellman assumption; LWE: learning with error problem; RLWE: Ring LWE.}}

			\end{tabular}
			\label{tab1}
		\end{center}
		
	\end{table*}
} 

\section{Preliminaries}
\subsection{Public key encryption with equality test supporting flexible authorization (PKEET-FA)}
In this section, we will recall the model of PKEET-FA and its security model.

We remark that a PKEET-FA system is a multi-user setting. Hence we assume that in our system throughout the paper, each user is assigned with an index $i$ with $1\leq i\leq N$ where $N$ is the number of users in the system.

\begin{defn}[PKEET-FA] \label{def:PKEET}
	PKEET-FA consists of the following polynomial-time algorithms:
	\begin{itemize}
		\item $\Setup(\lambda)$: On input a security parameter $\lambda$ and set of parameters, it outputs the a pair of a user's public key $\PK$ and secret key $\SK$.
		\item $\Enc(\PK,\bf{m})$: On input the public key $\PK$ and a message $\bf{m}$, it outputs a ciphertext $\CT$.
		\item $\Dec(\SK,\CT)$: On input the secret key $\SK$ and a ciphertext $\CT$, it outputs a message $\bf{m}'$ or $\perp$.
	\end{itemize}
	Suppose that the receiver $U_i$ (resp. $U_j$) has a public/secret key pair $(\PK_i,\SK_i)$ (resp. $(\PK_j,\SK_j)$), whose ciphertext is $\CT_i$ (resp. $\CT_j$). To realize Type-$\a$ $(\a = 1, 2, 3)$ authorization for $U_i$ and $U_j$ , algorithm $\Td_{\a}$ $(\a = 1, 2, 3)$ is defined to generate trapdoor for $U_i$'s ciphertext (or ciphertexts) that is (are) required to be compared with $U_j$'s ciphertext (or ciphertexts), and $\Test_{\a}$ $(\a = 1, 2, 3)$ algorithm to determine whether two receivers' ciphertexts contain the same message or not.\\
	Type-1 Authorization:
	\begin{itemize}
		\item $\Td_1(\SK_i)$: On input the secret key $\SK_i$ of the user $U_i$, it outputs a trapdoor $\td_{1,i}$ for $U_i$.
		\item $\Test(\td_{1,i},\td_{1,j},\CT_i,\CT_j)$: On input two trapdoors $\td_{1,i}, \td_{1,j}$ and two ciphertexts $\CT_i, \CT_j$ for users $U_i$ and $U_j$ respectively, it outputs $1$ or $0$.
	\end{itemize}
	Type-2 Authorization:
	\begin{itemize}
		\item $\Td_2(\SK_i, \CT_i)$: On input the secret key $\SK_i$ and ciphertext $\CT_i$ of the user $U_i$, it outputs a trapdoor $\td_{2,i}$ for $(U_i, \CT_i)$.
		\item $\Test(\td_{2,i},\td_{2,j},\CT_i,\CT_j)$: On input two trapdoors $\td_{2,i}, \td_{2,j}$ and two ciphertexts $\CT_i, \CT_j$ for users $U_i$ and $U_j$ respectively, it outputs $1$ or $0$.
	\end{itemize}
	Type-3 Authorization:
	\begin{itemize}
		\item $\Td_{3,i}(\SK_i, \CT_i)$: On input the secret key $\SK_i$ and ciphertext $\CT_i$ of the user $U_i$, it outputs a trapdoor $\td_{3, i}$ for $(U_i, \CT_i)$.
		\item $\Td_{3, j}(\SK_j)$: On input the secret key $\SK_j$ of the user $U_j$, it outputs a trapdoor $\td_{3, j}$ for $U_j$.
		\item $\Test(\td_{3,i},\td_{3, j},\CT_i,\CT_j)$: On input two trapdoors $\td_{3, i}, \td_{3, j}$ and two ciphertexts $\CT_i, \CT_j$ for users $U_i$ and $U_j$ respectively, it outputs $1$ or $0$.
	\end{itemize}
\end{defn}
\begin{rem}
	Type-3 authorization is a combination of Type-1 authorization and Type-2 authorization, which is for comparing a single ciphertext of $U_i$ with all ciphertexts of $U_j$.
\end{rem}

\noindent\textbf{Correctness.} We say that a PKEET scheme is \textit{correct} if the following three condition hold:
\begin{description}
	\item[(1)] For any security parameter $\lambda$, any user $U_i$ and any message $\bf{m}$, it holds that
	$$\Pr\left[ {\begin{gathered}
		\Dec(\SK_i,\CT_i)=\bf{m}\end{gathered}  
		\left| \begin{gathered}
		(\PK_i,\SK_i)\gets\Setup(\lambda)\\
		\CT_i\gets\Enc(\PK_i,\bf{m})
		\end{gathered}  \right.} \right]=1.$$
	\item[(2)] For any security parameter $\lambda$, any users $U_i$, $U_j$ and any messages $\bf{m}_i, \bf{m}_j$, it holds that:    
	$$\Pr\left[{
		\Test\left( \begin{gathered}
		\td_{\a, i} \\
		\td_{\a, j} \\
		\CT_i \\
		\CT_j \\ 
		\end{gathered}  \right) = 1\left| \begin{array}{l}
		(\PK_i,\SK_i)\gets\Setup(\lambda) \\
		\CT_i\gets\Enc(\PK_i,\bf{m}_i) \\
		\td_{\a, i}\gets\Td_\a() \\
		(\PK_j,\SK_j)\gets\Setup(\lambda) \\
		\CT_j\gets\Enc(\PK_j,\bf{m}_j) \\
		\td_{\a, j}\gets\Td_{\a}() 
		\end{array}  \right.} \right]=1$$
	if $\bf{m}_i=\bf{m}_j$ regardless of whether $i=j$, where $\a = 1, 2, 3$.
	
	\item[(3)] For any security parameter $\lambda$, any users $U_i$, $U_j$ and any messages $\bf{m}_i, \bf{m}_j$, it holds that
	$$\Pr\left[{
		\Test\left( \begin{gathered}
		\td_{\a, i} \\
		\td_{\a, j} \\
		\CT_i \\
		\CT_j \\ 
		\end{gathered}  \right) = 1\left| \begin{array}{l}
		(\PK_i,\SK_i)\gets\Setup(\lambda) \\
		\CT_i\gets\Enc(\PK_i,\bf{m}_i) \\
		\td_{\a, i}\gets\Td_{A} \\
		(\PK_j,\SK_j)\gets\Setup(\lambda) \\
		\CT_j\gets\Enc(\PK_j,\bf{m}_j) \\
		\td_{\a, j}\gets\Td_{\a}() 
		\end{array}  \right.} \right]$$
	is negligible in $\lambda$ for any ciphertexts $\CT_i$, $\CT_j$ such that $\Dec(\SK_i,\CT_i)\ne\Dec(\SK_j,\CT_j)$ regardless of whether $i=j$, where $\a = 1, 2, 3$.
\end{description}

\noindent\textbf{Security model of PKEET.} Because the trapdoor algorithm for Type-3 authorization can be obtained by the combination of those for Type-1 and Type-2 authorizations, we omit Type-3 authorization queries and only provide Type-$\a$ $(\a = 1,2)$ authorization queries to the adversary in the security games for simplicity. For the security model of PKEET-FA, we consider two types of adversaries:
\begin{itemize}
	\item[$\bullet$] Type-I adversary: for this type, the adversary can request to issue a trapdoor for authorization of the target user and thus can perform authorization (equality) tests on the challenge ciphertext. The aim of this type of adversaries is to reveal the message in the challenge ciphertext.
	\item[$\bullet$] Type-II adversary: for this type, the adversary cannot request to issue a trapdoor for authorization of the target user and thus cannot perform equality tests on the challenge ciphertext. The aim of this type of adversaries is to distinguish which message is in the challenge ciphertext between two candidates.
\end{itemize}
The security model of a PKEET scheme against two types of adversaries above is described in the following.\\

\noindent\textbf{OW-CCA security against Type-I adversaries.} We illustrate the game between a challenger $\cal{C}$ and a Type-I adversary $\cal{A}$ who can have a trapdoor for all ciphertexts of the target user, say $U_{\theta}$, that he wants to attack, as follows:
\begin{enumerate}
	\item \textbf{Setup:} The challenger $\cal{C}$ runs $\Setup(\lambda)$ to generate the key pairs $(\PK_i,\SK_i)$ for all users with $i=1,\cdots,N$, and gives $\{\PK_i\}_{i=1}^N$ to $\cal{A}$.
	\item \textbf{Phase 1:}  The adversary $\cal{A}$ may make queries polynomially many times adaptively and in any order to the following oracles:
	\begin{itemize}
		\item $\cal{O}^{\SK}$: an oracle that on input an index $i$ (different from $\theta$), returns the $U_i$'s secret key $\SK_i$.
		\item $\cal{O}^\Dec$: an oracle that on input a pair of an index $i$ and a ciphertext $\CT_i$, returns the output of $\Dec(\SK_i,\CT_i)$ using the secret key of the user $U_i$.
		\item $\cal{O}^{\Td_\a}$ for $\a = 1, 2$: 
		\begin{itemize}
			\item on input $i$, $\cal{O}^{\Td_1}$ returns $\td_i$ by running $\td_{1,i}\gets\Td(\SK_i)$.
			\item On input $(i, \CT_i)$, $\cal{O}^{\Td_2}$ returns by running $\td_{2,i} \gets \Td_2(\SK_i, \CT_i)$.
		\end{itemize}
	\end{itemize}
	\item \textbf{Challenge:} $\cal{C}$ chooses a random message $\bf{m}$ in the message space and run $\CT_{\theta}^*\gets\Enc(\PK_{\theta},\bf{m})$, and sends $\CT_{\theta}^*$ to $\cal{A}$.
	\item \textbf{Guess:} $\cal{A}$ output $\bf{m}'$.
\end{enumerate}
The adversary $\cal{A}$ wins the above game if $\bf{m}=\bf{m}'$ and the success probability of $\cal{A}$ is defined as
$$\Adv^{\OWa, Type-\a}_{\cal{A},\text{PKEET}}(\lambda):=\Pr[\bf{m}=\bf{m}'].$$
\begin{rem}
	If the message space is polynomial in the security parameter or the min-entropy of the message distribution is much lower than the security parameter then a Type-I adversary $\cal{A}$ with a trapdoor for the challenge ciphertext can reveal the message in polynomial-time or small exponential time in the security parameter, by performing the equality tests with the challenge ciphertext and all other ciphertexts of all messages generated by himself. Hence to prevent this attack, we assume that the size of the message space $\cal{M}$ is exponential in the security parameter and the min-entropy of the message distribution is sufficiently higher than the security parameter.
\end{rem}

\noindent\textbf{IND-CCA security against Type-II adversaries.} We present the game between a challenger $\cal{C}$ and a Type-II adversary $\cal{A}$ who cannot have a trapdoor for all ciphertexts of the target user $U_{\theta}$ as follows:
\begin{enumerate}
	\item \textbf{Setup:} The challenger $\cal{C}$ runs $\Setup(\lambda)$ to generate the key pairs $(\PK_i,\SK_i)$ for all users with $i=1,\cdots,N$, and gives $\{\PK_i\}_{i=1}^N$ to $\cal{A}$.
	\item \textbf{Phase 1:}  The adversary $\cal{A}$ may make queries polynomially many times adaptively and in any order to the following oracles:
	\begin{itemize}
		\item $\cal{O}^{\SK}$: an oracle that on input an index $i$ (different from $\theta$), returns the $U_i$'s secret key $\SK_i$.
		\item $\cal{O}^\Dec$: an oracle that on input a pair of an index $i$ and a ciphertext $\CT_i$, returns the output of $\Dec(\SK_i,\CT_i)$ using the secret key of the user $U_i$.
		\item $\cal{O}^{\Td_\a}$ for $\a = 1, 2$: 
		\begin{itemize}
			\item on input $i (\neq \theta)$, $\cal{O}^{\Td_1}$ returns $\td_i$ by running $\td_{1,i}\gets\Td(\SK_i)$.
			\item On input $(i (\neq \theta), \CT_i)$, $\cal{O}^{\Td_2}$ returns by running $\td_{2,i} \gets \Td_2(\SK_i, \CT_i)$.
		\end{itemize}
	\end{itemize}
	\item \textbf{Challenge:} $\cal{A}$ chooses message $\bf{m}$ and pass to $\cal{C}$, who then selects a random bit $r\in\{0,1\}$. If $r= 0$, runs $\CT^*_{\theta, r}\gets\Enc(\PK_{\theta},\bf{m})$; otherwise chose a random ciphertext as $\CT^*_{\theta, r}$  and sends $\CT^*_{\theta,r}$ to $\cal{A}$.
	\item \textbf{Guess:} $\cal{A}$ output $r'$.
\end{enumerate}
The adversary $\cal{A}$ wins the above game if $r=r'$ and the advantage of $\cal{A}$ is defined as
$$\Adv_{\cal{A},\text{PKEET}}^{\INDa, Type-\a}:=\left|\Pr[r=r']-\frac{1}{2}\right|.$$


\subsection{Lattices}
Lattices are discrete subgroups of $\ZZ^m$. Specially, a lattice $\Lambda$ in $\ZZ^m$ with basis $B=[\b_1,\cdots,\b_n]\in\ZZ^{m\times n}$, where each $\b_i$ is written in column form, is defined as
$$\Lambda:=\left\{\sum_{i=1}^n\b_ix_i | x_i\in\ZZ~\forall i=1,\cdots,n \right\}\subseteq\ZZ^m.$$
We call $n$ the rank of $\L$ and if $n=m$ we say that $\L$ is a full rank lattice. In this paper, we mainly consider full rank lattices containing $q\ZZ^m$, called $q$-ary lattices, defined as the following, for a given matrix $A\in\ZZ^{n\times m}$ and $\bf{u}\in\ZZ_q^n$
\begin{align*}
\L_q(A) &:= \left\{ \e\in\ZZ^m ~\rm{s.t.}~ \exists \bf{s}\in\ZZ_q^n~\rm{where}~A^T\bf{s}=\bf{e}\mod q \right\}\\
\Lp_q(A) &:= \left\{ \e\in\ZZ^m~\rm{s.t.}~A\e=0\mod q \right\} \\
\L_q^{\bf{u}}(A) &:=  \left\{ \e\in\ZZ^m~\rm{s.t.}~A\e=\bf{u}\mod q \right\}
\end{align*}
Note that if $\bf{t}\in\L_q^{\bf{u}}(A)$ then $\L_q^{\bf{u}}(A)=\Lp_q(A)+\bf{t}$.

Let $S=\{\s_1,\cdots,\s_k\}$ be a set of vectors in $\mathbb{R}^m$. We denote by $\|S\|:=\max_i\|\s_i\|$ for $i=1,\cdots,k$, the maximum $l_2$ length of the vectors in $S$. We also denote $\tilde{S}:=\{\tilde{\s}_1,\cdots,\tilde{\s}_k \}$ the Gram-Schmidt orthogonalization of the vectors $\s_1,\cdots,\s_k$ in that order. We refer to $\|\tilde{S}\|$ the Gram-Schmidt norm of $S$. 

Ajtai~\cite{Ajtai99} first proposed how to sample a uniform matrix $A\in\ZZ_q^{n\times m}$ with an associated basis $S_A$ of $\Lp_q(A)$ with low Gram-Schmidt norm. It is improved later by Alwen and Peikert~\cite{AP09} in the following Theorem.

\begin{thm}\label{thm:TrapGen}
	Let $q\geq 3$ be odd and $m:=\lceil 6n\log q\rceil$. There is a probabilistic polynomial-time algorithm $\TrapGen(q,n)$ that outputs a pair $(A\in\ZZ_q^{n\times m},S\in\ZZ^{m\times m})$ such that $A$ is statistically close to a uniform matrix in $\ZZ_q^{n\times m}$ and $S$ is a basis for $\Lp_q(A)$ satisfying
	\[\|\tilde{S}\|\leq O(\sqrt{n\log q})\quad\text{and}\quad\|S\|\leq O(n\log q)\]
	with all but negligible probability in $n$.
\end{thm}

\begin{defn}[Gaussian distribution]
	Let $\L\subseteq\ZZ^m$ be a lattice. For a vector $\bf{c}\in\RR^m$ and a positive parameter $\sigma\in\RR$, define:
	$$\rho_{\sigma,\c}(\x)=\exp\left(\pi\frac{\|\x-\c\|^2}{\sigma^2}\right)\quad\text{and}\quad
	\rho_{\sigma,\c}(\L)=\sum_{\x\in\L}\rho_{\sigma,\c}(\x).    $$
	The discrete Gaussian distribution over $\L$ with center $\c$ and parameter $\sigma$ is
	$$\forall \bf{y}\in\L\quad,\quad\cal{D}_{\L,\sigma,\c}(\bf{y})=\frac{\rho_{\sigma,\c}(\bf{y})}{\rho_{\sigma,\c}(\L)}.$$
\end{defn}
For convenience, we will denote by $\rho_\sigma$ and $\cal{D}_{\L.\sigma}$ for $\rho_{\bf{0},\sigma}$ and $\cal{D}_{\L,\sigma,\bf{0}}$ respectively. When $\sigma=1$ we will write $\rho$ instead of $\rho_1$. We can extend this definition to a positive definite covariance matrix $\Sigma=BB^T$: $\rho_{\sqrt{\Sigma},\bf{c}}(\bf{x})=\exp(-\pi(\bf{x}-\bf{c})^T\Sigma^{-1}(\bf{x}-\bf{c}))$. It is well-known that for a vector $\bf{x}$ sampled in $D_{\ZZ^m,\sigma}$, one has that $\|\bf{x}\|\leq t\sigma\sqrt{m}$ with overwhelming probability.	

We recall below in Theorem~\ref{thm:Gauss} some useful result. The first one comes from~{\cite[Lemma 4.4]{MR04}} . The second one is from~\cite{CHKP10} and formulated in ~{\cite[Theorem 17]{ABB10-EuroCrypt}} and the last one is from~{\cite[Theorem 19]{ABB10-EuroCrypt}}.

\begin{thm}\label{thm:Gauss}
	Let $q> 2$ and let $A, B$ be a matrix in $\ZZ_q^{n\times m}$ with $m>n$ and $B$ is rank $n$. Let $T_A, T_B$ be a basis for $\Lp_q(A)$ and  $\Lp_q(B)$ respectively.
	Then for $c\in\RR^m$ and $U\in\ZZ_q^{n\times t}$:
	\begin{enumerate}
		\item Let $M$ be a matrix in $\ZZ_q^{n\times m_1}$ and $\sigma\geq\|\widetilde{T_A}\|\omega(\sqrt{\log(m+m_1)})$. Then there exists a PPT algorithm $\SampleLeft(A,M,T_A,U,\sigma)$ that outputs a matrix $\e\in\ZZ^{(m+m_1)\times t}$ distributed statistically close to $\cal{D}_{\L_q^{U}(F_1),\sigma}$ where $F_1:=(A~|~M)$. In particular $\e\in \L_q^{U}(F_1)$, i.e., $F_1\cdot\e=U\mod q$.
		
		In addition, if $A$ is rank $n$ then there is a PPT algorithm $\SampleBasisLeft(A,M,T_A,\sigma)$ that outputs a basis of $\Lambda_q^\perp(F_1)$.
		\item Let $R$ be a matrix in $\ZZ^{k\times m}$ and let $s_R:=\sup_{\|\x\|=1}\|R\x\|$. Let $F_2:=(A~|~AR+B)$. Then for  $\sigma\geq\|\widetilde{T_B}\|s_R\omega(\sqrt{\log m})$, there exists a PPT algorithm \\$\SampleRight(A,B,R,T_B,U,\sigma)$ that outputs a matrix $\e\in\ZZ^{(m+k)\times t}$ distributed statistically close to $\cal{D}_{\L_q^{U}(F_2),\sigma}$. In particular $\e\in \L_q^{U}(F_2)$, i.e., $F_2\cdot\e=U\mod q$. 
		
		Note that when $R$ is a random matrix in $\{-1,1\}^{m\times m}$ then $s_R<O(\sqrt{m})$ with overwhelming probability (cf.~{\cite[Lemma 15]{ABB10-EuroCrypt}}).
	\end{enumerate}
\end{thm}

The security of our construction reduces to the LWE (Learning With Errors) problem introduced by Regev~\cite{Regev05}.
\begin{defn}[LWE problem]
	Consider publicly a prime $q$, a positive integer $n$, and a distribution $\chi$ over $\ZZ_q$. An $(\ZZ_q,n,\chi)$-LWE problem instance consists of access to an unspecified challenge oracle $\cal{O}$, being either a noisy pseudorandom sampler $\cal{O}_\s$ associated with a secret $\s\in\ZZ_q^n$, or a truly random sampler $\cal{O}_\$$ who behaviors are as follows:
	\begin{description}
		\item[$\cal{O}_\s$:] samples of the form $(\u_i,v_i)=(\u_i,\u_i^T\s+x_i)\in\ZZ_q^n\times\ZZ_q$ where $\s\in\ZZ_q^n$ is a uniform secret key, $\u_i\in\ZZ_q^n$ is uniform and $x_i\in\ZZ_q$ is a noise withdrawn from $\chi$.
		\item[$\cal{O}_\$$:] samples are uniform pairs in $\ZZ_q^n\times\ZZ_q$.
	\end{description}\
	The $(\ZZ_q,n,\chi)$-LWE problem allows responds queries to the challenge oracle $\cal{O}$. We say that an algorithm $\cal{A}$ decides the $(\ZZ_q,n,\chi)$-LWE problem if 
	$$\Adv_{\cal{A}}^{\mathsf{LWE}}:=\left|\Pr[\cal{A}^{\cal{O}_\s}=1] - \Pr[\cal{A}^{\cal{O}_\$}=1] \right|$$    
	is non-negligible for a random $\s\in\ZZ_q^n$.
\end{defn}
Regev~\cite{Regev05} showed that (see Theorem~\ref{thm:LWE} below) when $\chi$ is the distribution $\overline{\Psi}_\alpha$ of the random variable $\lfloor qX\rceil\mod q$ where $\alpha\in(0,1)$ and $X$ is a normal random variable with mean $0$ and standard deviation $\alpha/\sqrt{2\pi}$ then the LWE problem is hard. 
\begin{thm}\label{thm:LWE}
	If there exists an efficient, possibly quantum, algorithm for deciding the $(\ZZ_q,n,\overline{\Psi}_\alpha)$-LWE problem for $q>2\sqrt{n}/\alpha$ then there is an efficient quantum algorithm for approximating the SIVP and GapSVP problems, to within $\tilde{\cal{O}}(n/\alpha)$ factors in the $l_2$ norm, in the worst case.
\end{thm}
Hence if we assume  the hardness of approximating the SIVP and GapSVP problems in lattices of dimension $n$ to within polynomial (in $n$) factors, then it follows from Theorem~\ref{thm:LWE} that deciding the LWE problem is hard when $n/\alpha$ is a polynomial in $n$.

In this paper, we will also deal with ideal lattices, i.e., lattices arising from polynomial rings. Specially, when $n$ is a power of two, we consider the ring $R:=\ZZ[x]/(x^n+1)$, the ring of integers of the cyclotomic number field $\QQ[x]/(x^n+1)$. The ring $R$ is isomorphic to the integer lattice $\ZZ^n$ through mapping a polynomial $f=\sum_{i=0}^{n-1}f_ix^i$ to its vector of coefficients $(f_0,f_1,\cdots,f_{n-1})$ in $\ZZ^n$. In this paper, we will consider such ring $R$ and denote by $R_q=R/qR=\ZZ_q[x]/(x^n+1)$ where $q=1\mod 2n$ is a prime. In this section, we follow~\cite{BertFRS18-implement} to recall some useful results in ideal lattices.

	We use ring variants LWE, proposed by~\cite{LPR10}, and proven to be as hard as the GapSVP/SIVP problems on ideal lattices.


\begin{defn}[$\RLWE$] Given $\bf{a}=(a_1,\cdots,a_m)^T\in R_q^m$ a vector of $m$ uniformly random polynomials, and $\bf{b}=\bf{a}s+\bf{e}$ where $s\hookleftarrow U(R_q)$ and $\bf{e}\hookleftarrow D_{R^m,\sigma}$, distinguish $(\bf{a},\bf{b}=\bf{a}s+\bf{e})$ from $(\bf{a},\bf{b})$ drawn from the uniform distribution over $R_q^m\times R_q^m$.
	
\end{defn}

In this paper, we use the ring version of trapdoors for ideal lattices, introduced in~\cite{MP12} and recently improved in~\cite{GM18}.
\begin{defn}
	Define $\bf{g}=(1,2,4,\cdots,2^{k-1})^T\in R_q$ with $k=\lceil\log_2(q)\rceil$ and call $\bf{g}$ the \textbf{gadget vector}, i.e., for which the inversion of $f_{\bf{g}^T}(\bf{z})=\bf{g}^T\bf{z}\in R_q$ is easy. The lattice $\L_q^\perp(\bf{g}^T)$ has a publicly known basis $B_q\in R^{k\times k}$ which satisfies that $\|\widetilde{B}_q\|\leq\sqrt{5}$.
\end{defn}

\begin{defn}[$\bf{g}$-trapdoor]
	 Let $\bf{a}\in R_q^m$ and $\bf{g}\in R_q^k$ with $k=\lceil\log_2(q)\rceil$ and $m>k$. A $\bf{g}$-trapdoor for $\bf{a}$ consist in a matrix of small polynomials $T\in R^{(m-k)\times k}$, following a discrete Gaussian distribution of parameter $\sigma$, such that 
	$$\bf{a}^T\left(\begin{array}{c}T \\ I_k \end{array} \right) =h\bf{g}^T$$
	for some invertible element $h\in R_q$. The polynomial $h$ is called the tag associated to $T$. The quality of the trapdoor is measured by its largest singular value $s_1(T)$.
\end{defn}
The Algorithm~\ref{algo:TrapGen} shows how to generate a random vector $\bf{a}\in R_q^m$ together with its trapdoor $T$.

\begin{algorithm}\label{algo:TrapGen}
	\DontPrintSemicolon
	
	\KwInput{the ring modulus $q$, a Gaussian parameter $\sigma$, optional $\bf{a}'\in R_q^{m-k}$ and $h\in R_q$. If no $\bf{a}', h$ are given, the algorithm chooses $\bf{a}'\hookleftarrow U(R_q^{m-k})$ and $h=1$.}
	\KwOutput{$\bf{a}\in R_q^m$ with its trapdoor $T\in R^{(m-k)\times k}$, of norm $\| T\|\leq t\sigma\sqrt{(m-k)n}$ associated to the tag $h$.}
	
	Choose $T\hookleftarrow D_{R^{(m-k)\times k},\sigma}, \bf{a}=(\bf{a}'^T|h\bf{g}-\bf{a}'^TT	)^T$.
	
	Return $(\bf{a},T)$.

	\caption{Algorithm  $\TrapGen(q,\sigma,\bf{a}',h)$}
\end{algorithm}

One of the main algorithm we use in our scheme is the preimage sampling algorithm $\SamplePre$, illustrated in Algorithm~\ref{algo:SamplePre}, which finds $\bf{x}$ such that $f_{\bf{a}^T}(\bf{x})=u$ for a given $u\in R_q$ and a public $\bf{a}\in R_q^m$ using the $\bf{g}$-trapdoor $T$ of $\bf{a}$, where $(\bf{a},T)\gets\TrapGen(q,\sigma,\bf{a}',h)$ as in Algorithm~\ref{algo:TrapGen}.

  \begin{algorithm}\label{algo:SamplePre}
 	\DontPrintSemicolon
 	
 	\KwInput{$\bf{a}\in R_q^m$, with its trapdoor $T\in R^{(m-k)\times k}$ associated to an invertible tag $h\in R_q$, $u\in R_q$ and $\zeta,\sigma,\alpha$ three Gaussian parameters}
 	\KwOutput{$\bf{x}\in R_q^m$ following a discrete Gaussian distribution of parameters $\zeta$ satisfying $\bf{a}^T\bf{x}=u\in R_q$.}
 	
 	$\bf{p}\gets\mathsf{SampleP}(q,\zeta,\alpha,T)$, $v\gets h^{-1}(u-\bf{a}^T\bf{p})$
 	
 	$\bf{z}\gets\mathsf{SamplePolyG}(\sigma,v)$, $\bf{x}\gets \bf{p}+\left(\begin{array}{c}T \\ I_k \end{array} \right)\bf{z}$
 	
 	\caption{Algorithm  $\mathsf{SamplePre}(T,\bf{a},h,\zeta,\sigma,\alpha,u)$}
 \end{algorithm}

The algorithm~$\SamplePre$ uses the following two algorithm:
\begin{itemize}
	\item The algorithm $\SampleP(q,\zeta,\alpha,T)$ on input the ring modulus $q$, $\zeta$ and $\alpha$ two Gaussian parameters and $T\hookleftarrow D_{R^{(m-k)\times k},\sigma}$, output $\bf{p}\hookleftarrow D_{R^m,\sqrt{\Sigma_{\bf{p}}}}$ where $\Sigma_\bf{p}=\zeta^2I_m-\alpha^2\left(\begin{array}{c}
	T \\ I_k
	\end{array}\right)(T^T I_k)$.
	\item The algorithm $\mathsf{SamplePolyG}$ on input a Gaussian parameter $\sigma$ and a target $v\in R_q$, outputs $\bf{z}\hookleftarrow D_{\Lp_q(\bf{g}^T),\alpha,v}$ with $\alpha=\sqrt{5}\sigma$.
\end{itemize}
For the special case of the cyclotomic number ring $R=\ZZ[x]/(x^n+1)$ in our paper, the algorithm ~$\SampleP$ can be efficiently implemented as in~{\cite[Section 4.1]{GM18}}. For algorithm~$\mathsf{SamplePolyG}$, one needs to call the algorithm~$\mathsf{SampleG}$ in~{\cite[Section 3.2]{GM18}} $n$ times.

Our construction follows the IBE construction in~\cite{ABB10-EuroCrypt}. In such a case, we need an encoding hash function $H:\ZZ_q^n\to R_q$ to map identities to $\ZZ_q^n$ to invertible elements in $R_q$; such an $H$ is called \textit{encoding with Full-Rank Differences (FRD)} in \cite{ABB10-EuroCrypt}. We require that $H$ satisfies the following properties:
\begin{itemize}
	\item for all distinct $u,v\in\ZZ_q^n$, the element $H(u)-H(v)\in R_q$ is invertible; and
	\item $H$ is computable in polynomial time (in $n\log(q)$).
\end{itemize}
To implement such an encoding $H$, there are several methods proposed in~\cite{GM18,DM14,LS18} and we refer to {\cite[Section 2.4]{GM18}} for more details.

\section{PKEET-FA Construction over Integer Lattices}
	\label{sec:PKEET-FA}
	\begin{description}
	\item[Setup($\lambda$)]$\\$ On input a security parameter $\lambda$, set the parameters $q,n,m,\sigma,\alpha$ as in section \ref{sec:params}
	\begin{enumerate}
		\item Use $\TrapGen(q,n)$ to generate uniformly random $n\times m$-matrices $A, A'\in\ZZ_q^{n\times m}$ together with trapdoors $T_{A}$ and $T_{A'}$ respectively.
		\item Select $l+1$ uniformly random $n\times m$ matrices $A_1,\cdots,A_l,B\in\ZZ_q^{n\times m}$.
		\item Let $H: \{0,1\}^*\to \{0,1\}^t$ and $H':\{0,1\}^*\to \{-1,1\}^l$ be hash functions.
		\item Select a uniformly random matrix $U\in\ZZ_q^{n\times t}$.
		\item Output the public key and the secret key
		$$\PK=(A,A',A_1,\cdots,A_l,B,U)~,~\SK=(T_A,T_{A'}).$$
	\end{enumerate} 
	
	\item[Encrypt($\PK,\bf{m}$)]$\\$ On input the public key $\PK$ and a message $\bf{m}\in\{0,1\}^t$, do:
	\begin{enumerate}
		\item Choose a uniformly random $\bf{s}_1, \bf{s}_2\in\ZZ_q^n$
		\item Choose $\bf{x}_1,\bf{x}_2\in\overline{\Psi}_\alpha^t$ and compute\footnote{Note that for a message $\bf{m}\in\{0,1\}^t$, we choose a random binary string $\bf{m}'$ of fixed length $t'$ large enough and by abusing of notation, we write $H(\bf{m})$ for $H(\bf{m}'\|\bf{m})$.}
		\begin{align*}
		\bf{c}_1 &= U^T\bf{s}_1 +\bf{x}_1 +\bf{m}\big\lfloor\frac{q}{2}\big\rfloor\\
		\bf{c}_2 &= U^T\bf{s}_2 +\bf{x}_2 +H(\bf{m})\big\lfloor\frac{q}{2}\big\rfloor \in\ZZ_q^t.
		\end{align*}
		
		\item Choose $\bf{b}\in\{-1,1\}^l$, and set 
		$$F_1=(A|B+\sum_{i=1}^lb_iA_i)~,~ F_2=(A'|B+\sum_{i=1}^lb_iA_i).$$
		\item Choose $l$ uniformly random matrices $R_i\in\{-1,1\}^{m\times m}$ for $i=1,\cdots,l$ and define $R=\sum_{i=1}^lb_iR_i\in\{-l,\cdots,l\}^{m\times m}$.
		\item Choose $\bf{y}_1, \bf{y}_2\in\overline{\Psi}_\alpha^m$ and set $\bf{z}_1=R^T\bf{y}_1, \bf{z}_2=R^T\bf{y}_2\in\ZZ_q^m$.
		\item Compute
		\begin{align*}
		\bf{c}_3&=F_1^T\bf{s}_1+[\bf{y}_1^T|\bf{z}_1^T]^T\in\ZZ_q^{2m},\\ \bf{c}_4&=F_2^T\bf{s}_2+[\bf{y}_2^T|\bf{z}_2^T]^T\in\ZZ_q^{2m}.
		\end{align*}
		
		\item The ciphertext is $$\CT=(\bf{b}, \bf{c}_1,\bf{c}_2,\bf{c}_3,\bf{c}_4)\in \{-1,1\}^l \times \ZZ_q^{2t+4m}.$$    
	\end{enumerate}
	
	\item[Decrypt($\PK,\SK,\CT$)]$\\$ On input public key $\PK$, private key $\SK$ and a ciphertext $\CT=(\bf{b}, \bf{c}_1,\bf{c}_2,\bf{c}_3,\bf{c}_4)$, do:
	\begin{enumerate}
		\item Sample $\bf{e}\in\ZZ^{2m\times t}$ from $$\bf{e}\gets\SampleLeft(A, B+\sum_{i=1}^lb_iA_i,T_A,U,\sigma).$$
		Note that $F_1\cdot\bf{e}=U$ in $\ZZ^{n\times t}_q$.
		\item Compute $\bf{w}\gets\bf{c}_1-\bf{e}^T\bf{c}_3\in\ZZ_q^t$.
		\item For each $i=1,\cdots, t$, compare $w_i$ and $\lfloor\frac{q}{2}\rfloor$. If they are close, output $m_i=1$ and otherwise output $m_i=0$. We then obtain the message $\bf{m}$.
		\item Sample $\bf{e}'\in\ZZ^{2m\times t}$ from
		$$\bf{e}'\gets\SampleLeft(A', B+\sum_{i=1}^lb_iA_i,T_{A'},U,\sigma).$$
		\item Compute $\bf{w}'\gets\bf{c}_2-(\bf{e}')^T\bf{c}_4\in\ZZ_q^t$.
		\item For each $i=1,\cdots,t$, compare $w'_i$ and $\lfloor\frac{q}{2}\rfloor$. If they are close, output $h_i=1$ and otherwise output $h_i=0$. We then obtain the vector $\bf{h}$.
		\item If $\bf{h}=H(\bf{m})$ then output $\bf{m}$, otherwise output $\perp$.
	\end{enumerate}
	
	Let $U_i$ and $U_j$ be two users of the system. We denote by $\CT_i = (\bf{c}_{i1},\bf{c}_{i2},\bf{c}_{i3},\bf{c}_{i4})$ (resp. $\CT_j = (\bf{c}_{j1},\bf{c}_{j2},\bf{c}_{j3},\bf{c}_{j 4})$) be a ciphertext of $U_i$ (resp. $U_j$).\\
	
	\item[Type-1 Authorization]~~
	
	\begin{itemize}
		\item $\Td_1$($\SK_i$): On input a user $\mathcal{U}_i$'s secret key $\SK_i=(K_{i,1}, K_{i,2})$, it outputs a trapdoor $\td_{1, i}=K_{i,2}$.
		\item Test($\td_{1, i},\td_{1, j},\CT_i,\CT_j$): On input trapdoors $\td_{1, i}, \td_{1, j}$ and ciphertexts $\CT_i,\CT_j$ for users $\mathcal{U}_i, \mathcal{U}_j$ respectively, computes 
		\begin{enumerate}
			\item For each $i$ (resp. $j$), do the following:
			\begin{itemize}
				\item $\bf{b}_i= (b_{i1},\cdots,b_{il})$ and sample $\bf{e}_i\in\ZZ^{2m\times t}$ from
				$$\bf{e_i}\gets\SampleLeft(A'_i, B_i+\sum_{k=1}^lb_{ik}A_{ik},T_{A'_i},U_i,\sigma).$$
				Note that $F_{i2}\cdot\bf{e}_i=U_i$ in $\ZZ^{n\times t}_q$.
				\item Compute $\bf{w}_i\gets \bf{c_{i2}}-\bf{e}_i^T\bf{c}_{i4}\in\ZZ_q^t$. For each $k=1,\cdots, t$, compare each coordinate $w_{ik}$ with $\lfloor\frac{q}{w}\rfloor$ and output $\bf{h}_{ik}=1$ if they are close, and $0$ otherwise. At the end, we obtain the vector $\bf{h}_i$ (resp. $\bf{h}_j$). 
			\end{itemize}
			\item Output $1$ if  $\bf{h}_i=\bf{h}_j$ and $0$ otherwise.
		\end{enumerate}
		
	\end{itemize}

	\item[Type-2 Authorization]:
	
	\begin{itemize}
		\item $\Td_2$($\SK_i, \CT_i$): On input a user $\mathcal{U}_i$'s secret key $\SK_i=(K_{i,1}, K_{i,2})$ and ciphertext $\CT_i$, it outputs a trapdoor $\td_{2, i}$ in following manner:\\
		$\bf{b}_i =(b_{i1},\cdots,b_{il})$ and sample $\td_{2, i} = \bf{e}_i\in\ZZ^{2m\times t}$ from
		$$\bf{e_i}\gets\SampleLeft(A'_i, B_i+\sum_{k=1}^lb_{ik}A_{ik},T_{A'_i},U_i,\sigma).$$
		Note that $F_{i2}\cdot\bf{e}_i=U_i$ in $\ZZ^{n\times t}_q$.
		\item Test($\td_{2, i},\td_{2, j},\CT_i,\CT_j$): On input trapdoors $\td_{2, i}, \td_{2, j}$ and ciphertexts $\CT_i,\CT_j$ for users $\mathcal{U}_i, \mathcal{U}_j$ respectively, computes 
		\begin{enumerate}
			\item For each $i$ (resp. $j$), do the following:
			Compute $\bf{w}_i\gets \bf{c_{i2}}-\bf{e}_i^T\bf{c}_{i4}\in\ZZ_q^t$. For each $k=1,\cdots, t$, compare each coordinate $w_{ik}$ with $\lfloor\frac{q}{w}\rfloor$ and output $\bf{h}_{ik}=1$ if they are close, and $0$ otherwise. At the end, we obtain the vector $\bf{h}_i$ (resp. $\bf{h}_j$). 
			\item Output $1$ if  $\bf{h}_i=\bf{h}_j$ and $0$ otherwise.
		\end{enumerate}
		
	\end{itemize}

	\item[Type-3 Authorization]:
	
	\begin{itemize}
		\item $\Td_{3, i}$($\SK_i, \CT_i$): On input a user $\mathcal{U}_i$'s secret key $\SK_i=(K_{i,1}, K_{i,2})$ and ciphertext $\CT_i$, it outputs a trapdoor $\td_{3, i}$ in following manner:\\
		$\bf{b}_i =(b_{i1},\cdots,b_{il})$ and sample $\td_{2, i} = \bf{e}_i\in\ZZ^{2m\times t}$ from
		$$\bf{e_i}\gets\SampleLeft(A'_i, B_i+\sum_{k=1}^lb_{ik}A_{ik},T_{A'_i},U_i,\sigma).$$
		Note that $F_{i2}\cdot\bf{e}_i=U_i$ in $\ZZ^{n\times t}_q$.
		\item $\Td_{3,j}$($\SK_j$): On input a user $U_j$'s secret key $\SK_j=(K_{j,1}, K_{j,2})$, it outputs a trapdoor $\td_{3, j}=K_{j,2}$.
		\item Test($\td_{3, i},\td_{3, j},\CT_i,\CT_j$): On input trapdoors $\td_{3, i}, \td_{3, j}$ and ciphertexts $\CT_i,\CT_j$ for users $\mathcal{U}_i, \mathcal{U}_j$ respectively, computes 
		\begin{enumerate}
			\item For each $i$, do the following:
			Compute $\bf{w}_i\gets \bf{c_{i2}}-\bf{e}_i^T\bf{c}_{i4}\in\ZZ_q^t$. For each $k=1,\cdots, t$, compare each coordinate $w_{ik}$ with $\lfloor\frac{q}{w}\rfloor$ and output $\bf{h}_{ik}=1$ if they are close, and $0$ otherwise. At the end, we obtain the vector $\bf{h}_i$. 
			\item $\bf{b}_j= (b_{i1},\cdots,b_{jl})$ and sample $\bf{e}_j\in\ZZ^{2m\times t}$ from
			$$\SampleLeft(A'_j, B_j+\sum_{k=1}^lb_{jk}A_{jk},T_{A'_j},U_j,\sigma).$$
			Note that $F_{j2}\cdot\bf{e}_j=U_j$ in $\ZZ^{n\times t}_q$.
			\item Compute $\bf{w}_j\gets \bf{c_{j2}}-\bf{e}_j^T\bf{c}_{j4}\in\ZZ_q^t$. For each $k=1,\cdots, t$, compare each coordinate $w_{jk}$ with $\lfloor\frac{q}{w}\rfloor$ and output $\bf{h}_{jk}=1$ if they are close, and $0$ otherwise. At the end, we obtain the vector $\bf{h}_j$.
			\item Output $1$ if  $\bf{h}_i=\bf{h}_j$ and $0$ otherwise.
		\end{enumerate}
		
	\end{itemize}
\end{description}

\begin{thm}
	Our PKEET-FA construction above is correct if $H$  is a collision-resistant hash function.
\end{thm}
\begin{proof}
	It is easy to see that if $\CT$ is a valid ciphertext of $\bf{m}$ then the decryption will always output $\bf{m}$. Moreover, if $\CT_i$ and $\CT_j$ are valid ciphertext of $\bf{m}$ and $\bf{m}'$ of user $U_i$ and $U_j$ respectively. Then the Test process of Type-$\a$ (for $\a = 1, 2, 3$) checks whether $H(\bf{m})=H(\bf{m}')$. If so then it outputs $1$, meaning that $\bf{m}=\bf{m}'$, which is always correct with overwhelming probability since $H$ is collision resistant. Hence our PKEET-FA described above is correct.
\end{proof}
	
	\subsection{Parameters}\label{sec:params}
	We follow~{\cite[Section 7.3]{ABB10-EuroCrypt}} for choosing parameters for our scheme. Now for the system to work correctly we need to ensure
	\begin{itemize}
		\item the error term in decryption is less than ${q}/{5}$ with high probability, i.e., $q=\Omega(\sigma m^{3/2})$ and $\alpha<[\sigma lm\omega(\sqrt{\log m})]^{-1}$,
		\item that the $\TrapGen$ can operate, i.e., $m>6n\log q$,
		\item that $\sigma$ is large enough for $\SampleLeft$ and $\SampleRight$, i.e.,
		$\sigma>lm\omega(\sqrt{\log m})$,
		\item that Regev's reduction applies, i.e., $q>2\sqrt{n}/\alpha$,
		\item that our security reduction applies (i.e., $q>2Q$ where $Q$ is the number of identity queries from the adversary).
	\end{itemize}
	Hence the following choice of parameters $(q,m,\sigma,\alpha)$ from \cite{ABB10-EuroCrypt} satisfies all of the above conditions, taking $n$ to be the security parameter:
	\begin{equation}\label{eq:params}
	\begin{aligned}
	& m=6n^{1+\delta}\quad,\quad q=\max(2Q,m^{2.5}\omega(\sqrt{\log n})) \\
	& \sigma = ml\omega(\sqrt{\log n})\quad,\quad\alpha=[l^2m^2\omega(\sqrt{\log n})]
	\end{aligned}
	\end{equation}
	and round up $m$ to the nearest larger integer and $q$ to the nearest larger prime. Here we assume that $\delta$ is such that $n^\delta>\lceil\log q\rceil=O(\log n)$.
	\subsection{Security analysis}
	In this section, we will prove that our proposed scheme is OW-CCA secure against Type-I adversaries (cf.~Theorem~\ref{thm:OW}) and IND-CCA secure against Type-II adversaries (cf. Theorem~\ref{thm:IND}).
	
	\begin{theorem}\label{thm:OW}
		The PKEET with parameters $(q,n,m,\sigma,\alpha)$ as in~\eqref{eq:params} is $\OWa$ secure provided that $H$ is a one-way hash function, $H'$ is a collision-resistant hash function, and Full-IBE of \cite{ABB10-EuroCrypt} is IND-CPA secure based on the hardness of $(\ZZ_q,n,\bar\Psi_\a)$-LWE assumption. 
\end{theorem}

\begin{proof}
We show $\OWa$ security in the standard model. Let $\mathcal{A}$ be a PPT adversary that attacks the PKEET-FA scheme and has advantage $\epsilon$. We will construct an adversary $\mathcal{S}$ that attacks the Full-IBE of \cite{ABB10-EuroCrypt} by simulating the view of $\mathcal{A}$, and has advantage $\epsilon$.

The adversary $\mathcal{S}$ works as follows:\\
Let $\mathcal{A}$ choose $\mathcal{U}_{\theta}$ as target user.\\
 On input  $\PK=(A,A_1,\cdots,A_l,B,U)$ of Full-IBE, $\mathcal{S}$ append $A'$ in $\PK$, generates $T_{A'}$, sets the PKEET-FA instance for $\mathcal{U}_{\theta}$, and simulates the view of $\mathcal{A}$ for $\mathcal{U}_{\theta}$.\\
 For other user $\mathcal{U}_i$, $\mathcal{S}$ sets $\PK=(A_i,A'_i,A_1,\cdots,A_l,B_i,U_i)$, and generates $T_{A_i}$, $T_{A'_i}$ by $\TrapGen(q,n)$ to simulates the view of $\mathcal{A}$ according to the scheme.\\

 \begin{enumerate}
 \item $Query~ Phase:$
 \begin{itemize}
                \item $Secret~ Key~ Query$: On $\mathcal{A}$'s query for the secret key of $i(\neq \theta)$-th user,   $\mathcal{S}$ generates $T_{A_i}$, $T_{A'_i}$ by $\TrapGen(q,n)$ and return to $\mathcal{A}$.
		\item $Decryption~ Query$: 
		\begin{itemize}
		    \item On $\mathcal{A}$'s query on a pair of an index $i(\neq \theta)$ and a ciphertext $\CT_i$, $\mathcal{S}$ returns the output of $\Dec(\SK_i,\CT_i)$ using the secret key of the user $U_i$.
		    \item $\mathcal{A}$'s query on a pair of an index $\theta$ and a ciphertext $\CT_{\theta}$, $\mathcal{S}$ makes secret key query to Full-IBE oracle for secret key for the identity $\bf{b}_{\theta}$ and receives $\bf{e}_{\theta}$. $\mathcal{S}$ decrypt $\CT_{\theta}$ according to $\Dec$ of PKEET-FA and returns to $\mathcal{A}$.
		\end{itemize}
		\item $\Td_1~ Query$: On $\mathcal{A}$'s query on an index $i$, $\mathcal{S}$ returns $\td_{1,i}\gets\Td(\SK_i)$.
		\item $\Td_2~ Query$: On $\mathcal{A}$'s query on a pair of an index $i$ and a ciphertext $\CT_i$, $\mathcal{S}$ returns $\td_{2,i} \gets \Td_2(\SK_i, \CT_i)$.
		 
		 \end{itemize}
	\item $Challenge~ Phase:$ $\mathcal{S}$ chooses $m$ from message space and sends $m, id^*(\neq b_{\theta})$ to the challenger of Full-IBE and Encrypt $\bf{H(m)}$ i.e., $\bf{c}_2,\bf{c}_4$. Challenger of Full-IBE choose $r \gets_\$ \{0, 1\}$. 
	\begin{itemize}
	    \item If $r = 0$,  Challenger of Full-IBE computes $Encrypt(\PK, id^*, m) = (\bf{c}_1,\bf{c}_3)$, and sends $(\bf{c}_1,\bf{c}_3)$ to $\mathcal{S}$.
	    \item If $r = 1$,  Challenger of Full-IBE chose $(\bf{c}_1,\bf{c}_3) \gets_\$ \ZZ_q^{t} \times \ZZ_q^{2m}$, and sends $(\bf{c}_1,\bf{c}_3)$ to $\mathcal{S}$.
	\end{itemize}
	After getting reply, $\mathcal{S}$ will send $(id^*, \bf{c}_1,\bf{c}_2, \bf{c}_3,\bf{c}_4)$ to $\mathcal{A}$.
	\item $Decision~ Phase:$ 
	\begin{itemize}
	    \item If $\mathcal{A}$ sends $m$ to $\mathcal{S}$, $\mathcal{S}$ will sends $r = 0$ to the Challenger of Full-IBE.
	    \item If $\mathcal{A}$ sends $\perp$ to $\mathcal{S}$, $\mathcal{S}$ will sends $r = 1$ to the Challenger of Full-IBE.
	\end{itemize}
\end{enumerate}

So, the simulation of view of $\mathcal{A}$ is perfect with the real scheme. So, the winning probability of $\mathcal{A}$ against PKEET-FA is same with the winning probability of $\mathcal{S}$ against Full-IBE.
Hence, the proof.
\end{proof}
		
	\begin{theorem}\label{thm:IND}
		The PKEET with parameters $(q,n,m,\sigma,\alpha)$ as in~\eqref{eq:params}  is $\INDa$ secure provided that $H'$ is a collision-resistant hash function, and Full-IBE of \cite{ABB10-EuroCrypt} is IND-CPA secure based on the hardness of $(\ZZ_q,n,\bar\Psi_\a)$-LWE assumption. 
	\end{theorem}
	
\begin{proof}
We show $\INDa$ security in the standard model. Let $\mathcal{A}$ be a PPT adversary that attacks the PKEET-FA scheme and has advantage $\epsilon$. We will construct an adversary $\mathcal{S}$ that attacks the Full-IBE of \cite{ABB10-EuroCrypt} by simulating the view of $\mathcal{A}$, and has advantage $\epsilon$.

The adversary $\mathcal{S}$ works as follows:\\
Let $\mathcal{A}$ choose $\mathcal{U}_{\theta}$ as target user.\\
 On input  $\PK=(A,A_1,\cdots,A_l,B,U)$ of Full-IBE, $\mathcal{S}$ append $A'$ in $\PK$, generates $T_{A'}$, sets the PKEET-FA instance for $\mathcal{U}_{\theta}$, and simulates the view of $\mathcal{A}$ for $\mathcal{U}_{\theta}$.\\
 For other user $\mathcal{U}_i$, $\mathcal{S}$ sets $\PK=(A_i,A'_i,A_1,\cdots,A_l,B_i,U_i)$, and generates $T_{A_i}$, $T_{A'_i}$ by $\TrapGen(q,n)$ to simulates the view of $\mathcal{A}$ according to the scheme.\\

 \begin{enumerate}
 \item $Query~ Phase:$
 \begin{itemize}
                \item $Secret~ Key~ Query$: On $\mathcal{A}$'s query for the secret key of $i(\neq \theta)$-th user,   $\mathcal{S}$ generates $T_{A_i}$, $T_{A'_i}$ by $\TrapGen(q,n)$ and return to $\mathcal{A}$.
		\item $Decryption~ Query$: 
		\begin{itemize}
		    \item On $\mathcal{A}$'s query on a pair of an index $i(\neq \theta)$ and a ciphertext $\CT_i$, $\mathcal{S}$ returns the output of $\Dec(\SK_i,\CT_i)$ using the secret key of the user $U_i$.
		    \item $\mathcal{A}$'s query on a pair of an index $\theta$ and a ciphertext $\CT_{\theta}$, $\mathcal{S}$ makes secret key query to Full-IBE oracle for secret key for the identity $\bf{b}_{\theta}$ and receives $\bf{e}_{\theta}$. $\mathcal{S}$ decrypt $\CT_{\theta}$ according to $\Dec$ of PKEET-FA and returns to $\mathcal{A}$.
		\end{itemize}
		\item $\Td_1~ Query$: On $\mathcal{A}$'s query on an index $i(\neq \theta)$, $\mathcal{S}$ returns $\td_{1,i}\gets\Td(\SK_i)$.
		\item $\Td_2~ Query$: On $\mathcal{A}$'s query on a pair of an index $i(\neq \theta)$ and a ciphertext $\CT_i$, $\mathcal{S}$ returns $\td_{2,i} \gets \Td_2(\SK_i, \CT_i)$.
		 
		 \end{itemize}
	\item $Challenge~ Phase:$ $\mathcal{S}$ receives $m$ from $\mathcal{A}$. $\mathcal{S}$ sends $m, id^*(\neq b_{\theta})$ to the challenger of Full-IBE and Encrypt $\bf{H(m)}$ i.e., $\bf{c}_2,\bf{c}_4$. Challenger of Full-IBE choose $r \gets_\$ \{0, 1\}$. 
	\begin{itemize}
	    \item If $r = 0$,  Challenger of Full-IBE computes $Encrypt(\PK, id^*, m) = (\bf{c}_1,\bf{c}_3)$, and sends $(\bf{c}_1,\bf{c}_3)$ to $\mathcal{S}$.
	    \item If $r = 1$,  Challenger of Full-IBE chose $(\bf{c}_1,\bf{c}_3) \gets_\$ \ZZ_q^{t} \times \ZZ_q^{2m}$, and sends $(\bf{c}_1,\bf{c}_3)$ to $\mathcal{S}$.
	\end{itemize}
	After getting reply, $\mathcal{S}$ will send $(id^*, \bf{c}_1,\bf{c}_2, \bf{c}_3,\bf{c}_4)$ to $\mathcal{A}$.
	\item $Decision~ Phase:$ 
	 $\mathcal{S}$ will receive $r'$ from $\mathcal{A}$ and passes to the challenger of Full-IBE.
\end{enumerate}

So, the simulation of view of $\mathcal{A}$ is perfect with the real scheme. So, the winning probability of $\mathcal{A}$ against PKEET-FA is same with the winning probability of $\mathcal{S}$ against Full-IBE.
Hence, the proof.
\end{proof}
	

\section{Instantiation from Lee et al. \cite{Lee2016}}\label{sec:construct Lee}
\begin{description}
	\item[$\Setup$($\lambda$)]$\\$ On input security parameter $\lambda$, and a maximum hierarchy depth $2$, set the parameters $q, n, m, {\bar \sigma}, {\bar \alpha}$. The vector ${\bar \sigma} ~\&~ {\bar \alpha} \in \mathbb{R}^2$ and we use $\sigma_l$ and $\alpha_l$ to refer to their $l$- th coordinate. 
	\begin{enumerate}
		\item Use algorithm $\TrapGen(q, n)$ to select uniformly random $n \times m$- matrices $A, A' \in \mathbb{Z}_q^{n \times m}$ with a basis $T_{A}, T_{A'}$ for $\L^{\perp}_q (A), \L^{\perp}_q (A')$, respectively. Repeat this Step until $A, A'$ has rank $n$.
		\item Select $3$ uniformly random $m \times m$ matrices $A_1, A_2,  B \in \mathbb{Z}_q^{n \times m}$.
		\item Select a uniformly random matrix $U\in\mathbb{Z}_q^{n\times t}$.
		\item We need some hash functions $H: \{0, 1\}^* \rightarrow\{0,1\}^t$,  $H_1: \{0, 1\}^* \rightarrow \{-1, 1\}^t$, $H_2 :\{0,1\}^*\to\ZZ_q$ and a full domain difference map $H' :\ZZ_q^n\to\ZZ_q^{n\times n}$ as in~{\cite[Section 5]{ABB10-EuroCrypt}}.
		\item Output the public key and the secret key
		$$\PK=(A,A_1,A_2,B,U)\quad,\quad\SK=(T_A, T_{A'}).$$
	\end{enumerate}

	\item[$\Enc(\PK,\bf{m})$] $\\$
	On input the public key $\PK$ and a message $\bf{m}\in\{0,1\}^t$ do
	\begin{enumerate}
		\item Choose uniformly random $\bf{s}_1,\bf{s}_2\in\ZZ_q^n$.
		\item Choose $\bf{x}_1,\bf{x}_2\in\overline{\Psi}_\alpha^t$ and compute
		\begin{align*}
		\bf{c}_1 &= U^T\bf{s}_1 +\bf{x}_1 +\bf{m}\big\lfloor\frac{q}{2}\big\rfloor\in\ZZ_q^t,\\
		\bf{c}_2 &= U^T\bf{s}_2 +\bf{x}_2 +H(\bf{m})\big\lfloor\frac{q}{2}\big\rfloor \in\ZZ_q^t.
		\end{align*}
		\item Select uniformly random $l-2$ matrices $A_3,\cdots,A_l\in\ZZ_q^{n\times m}$ and set
		$vk=A_3\|\cdots\|A_l$.
		\item Set $id := H_2(vk)\in\ZZ_q^n$.
		\item Build the following matrices in $\ZZ_q^{n\times 3m}$:
		\begin{align*}
		F_1 &= (A | A_1 + H'(0)\cdot B | A_2 + H'(id)\cdot B),\\
		F_2 &= (A | A_1 + H'(1)\cdot B | A_2 + H'(id)\cdot B).
		\end{align*}
		
		\item Choose a uniformly random $n\times 2m$ matrix $R$ in $\{-1,1\}^{n\times 2m}$.
		\item Choose $\bf{y}_1, \bf{y}_2\in\overline{\Psi}_\alpha^m$ and set $\bf{z}_1=R^T\bf{y}_1, \bf{z}_2=R^T\bf{y}_2\in\ZZ_q^{2m}$.
		\item Compute
		\begin{align*}
		\bf{c}_3&=F_1^T\bf{s}_1+[\bf{y}_1^T|\bf{z}_1^T]^T\in\ZZ_q^{3m},\\ \bf{c}_4&=F_2^T\bf{s}_2+[\bf{y}_2^T|\bf{z}_2^T]^T\in\ZZ_q^{3m}.
		\end{align*}
		\item Let $\bf{b}:= H_1(\bf{c}_1\|\bf{c}_2\|\bf{c}_3\|\bf{c}_4)\in\{-1,1\}^l$ and define a matrix
		$$F=(A'| B+\sum_{i=1}^lb_iA_i)\in\ZZ_q^{n\times 2m}.$$
		\item Extract a signature $\bf{e'}\in\ZZ^{2m\times t}$ by
		$$\bf{e'}\gets\SampleLeft(A',B+\sum_{i=1}^lb_iA_i,T_{A'},0,\sigma).$$
		Note that $F\cdot\bf{e'} =0\mod q$.
		\item Output the ciphertext 
		$$\CT=(vk,\bf{c}_1,\bf{c}_2,\bf{c}_3,\bf{c}_4,\bf{e'}).$$
	\end{enumerate}

	\item[$\Dec (\SK,\CT)$]$\\$ On input a secret key $\SK$ and a ciphertext $\CT$, do 
	\begin{enumerate}
		\item Parse the ciphertext $\CT$ into $$(vk,\bf{c}_1,\bf{c}_2,\bf{c}_3,\bf{c}_4,\bf{e}).$$
		\item Let $\bf{b}:= H_1(\bf{c}_1\|\bf{c}_2\|\bf{c}_3\|\bf{c}_4)\in\{-1,1\}^l$ and define a matrix
		$$F=(A'| B+\sum_{i=1}^lb_iA_i)\in\ZZ_q^{n\times 2m}.$$
		\item If $F\cdot\bf{e'}=0$ in $\ZZ_q$ and $\|\bf{e}\|\leq\sigma\sqrt{2m}$ then continue to Step 4; otherwise output $\perp$.
		\item Set $id := H_2(vk)\in\ZZ_q^n$ and build the following matrices in $\ZZ_q^{n\times 3m}$:
		\begin{align*}
		F_1 &= (A | A_1 + H'(0)\cdot B | A_2 + H'(id)\cdot B),\\
		F_2 &= (A | A_1 + H'(1)\cdot B | A_2 + H'(id)\cdot B).
		\end{align*}
		\item Use the secret key $\SK=T_A$ and the algorithm $\SampleLeft$ to sample matrices $\bf{e}_1,\bf{e}_2\in\ZZ_q^{3m\times t}$ such that $F_1\cdot\bf{e}_1=U$ and $F_2\cdot\bf{e}_2=U$ in $\ZZ_q^{n\times t}$.
		\item Compute $\bf{w}\gets\bf{c}_1-\bf{e}_1^T\bf{c}_3\in\ZZ_q^t$.
		\item For each $i=1,\cdots, t$, compare $w_i$ and $\lfloor\frac{q}{2}\rfloor$. If they are close, output $m_i=1$ and otherwise output $m_i=0$. We then obtain the message $\bf{m}$.
		\item Compute $\bf{w}'\gets\bf{c}_2-\bf{e}_2^T\bf{c}_4\in\ZZ_q^t$.
		\item For each $i=1,\cdots,t$, compare $w'_i$ and $\lfloor\frac{q}{2}\rfloor$. If they are close, output $h_i=1$ and otherwise output $h_i=0$. We then obtain the vector $\bf{h}$.
		\item If $\bf{h}=H(\bf{m})$  then output $\bf{m}$, otherwise output $\perp$.
	\end{enumerate}
	
	\item[$\Td(\SK_i)$]$\\$
	On input the secret key $\SK_i$ of a user $U_i$, run
	$$\td_i\gets\SampleBasisLeft(A_i,B_i+H'(1)\cdot A_{i,1})$$
	and returns the trapdoor $\td_i\in\ZZ_q^{2m\times 2m}$. Note that 
	$$(A_i\| B_i+H'(1)\cdot A_{i,1})\cdot\td_i =0\in\ZZ_q^{n\times 2m}.$$
	\item[$\Test(\td_i,\td_j,\CT_i,\CT_j)$]$\\$ On input trapdoors $\td_i,\td_j$ and ciphertexts $\CT_i,\CT_j$ of users $U_i$ and $U_j$ respectively, for $k=i,j$, do the following
	\begin{enumerate}
		\item Parse $\CT_k$ into 
		$$(vk_k,\bf{c}_{k,1},\bf{c}_{k,2},\bf{c}_{k,3},\bf{c}_{k,4},\bf{e}_k).$$
		\item Set $\bar A_k =(A_k|B_k+H'(1)\cdot A_{k,1})$. Sample $\bf{e}_k\in\ZZ_q^{3m\times t}$ from
		$$\SampleLeft(\bar A_k,B_k+H'(id_k)\cdot A_{k,2},\td_k,U,\sigma).$$
		\item Use $\bf{e}_k$ to decrypt $\bf{c}_{k,2}$, $\bf{c}_{k,4}$ as in Step 8-9 of $\Dec(\SK,\CT)$ above to obtain the hash value $\bf{h}_k$.
		\item If $\bf{h}_i=\bf{h}_j$ then ouput $1$; otherwise output $0$.
	\end{enumerate}
\end{description}

\begin{thm}[Correctness]
	The above PKEET is correct if the hash function $H$ is collision resistant.
\end{thm}
\begin{proof}
	Since we employ the multi-bit HIBE and signature scheme from~\cite{ABB10-EuroCrypt}, their correctness follow from~\cite{ABB10-EuroCrypt}. The Theorem follows from~{\cite[Theorem 1]{Lee2016}}.
\end{proof}

\subsubsection*{Parameters}
We follow~{\cite[Section 8.3]{ABB10-EuroCrypt}} for choosing parameters for our scheme. Now for the system to work correctly we need to ensure
\begin{itemize}
	\item the error term in decryption is less than ${q}/{5}$ with high probability, i.e., $q=\Omega(\sigma m^{3/2})$ and $\alpha<[\sigma lm\omega(\sqrt{\log m})]^{-1}$,
	\item that the $\TrapGen$ can operate, i.e., $m>6n\log q$,
	\item that $\sigma$ is large enough for $\SampleLeft$ and $\SampleRight$, i.e.,
	$\sigma>lm\omega(\sqrt{\log m})$,
	\item that Regev's reduction applies, i.e., $q>2\sqrt{n}/\alpha$,
\end{itemize}
Hence the following choice of parameters $(q,m,\sigma,\alpha)$ from \cite{ABB10-EuroCrypt} satisfies all of the above conditions, taking $n$ to be the security parameter:
\begin{equation}\label{eq:params-HIBE}
\begin{aligned}
& m=6n^{1+\delta}\quad,\quad q=\max(2Q,m^{2.5}\omega(\sqrt{\log n})) \\
& \sigma = ml\omega(\sqrt{\log n})\quad,\quad\alpha=[l^2m^2\omega(\sqrt{\log n})]
\end{aligned}
\end{equation}
and round up $m$ to the nearest larger integer and $q$ to the nearest larger prime. Here we assume that $\delta$ is such that $n^\delta>\lceil\log q\rceil=O(\log n)$.

\begin{thm}
	The PKEET constructed in Section~\ref{sec:construct Lee} with paramaters as in~\eqref{eq:params-HIBE} is $\IND$ secure provided that $H_1$ is collision resistant. 
\end{thm}

\begin{proof}
	The HIBE is $\textsf{IND-sID-CPA}$ secure by~{\cite[Theorem 33]{ABB10-EuroCrypt}} and the signature is strongly unforgeable by~{\cite[Section 7.5]{ABB10-EuroCrypt}}. The result follows from {\cite[Theorem 2]{Lee2016}}.
\end{proof}

\begin{thm}[{\cite[Theorem 3]{Lee2016}}]
	The PKEET with parameters $(q,n,m,\sigma,\alpha)$ as in~\eqref{eq:params-HIBE} is $\OW$ provided that $H$ is one-way and  $H_1$ is collision resistant.
\end{thm}

\begin{proof}
	The HIBE is $\textsf{IND-sID-CPA}$ secure by~{\cite[Theorem 33]{ABB10-EuroCrypt}} and the signature is strongly unforgeable by~{\cite[Section 7.5]{ABB10-EuroCrypt}}. The result follows from {\cite[Theorem 3]{Lee2016}}.
\end{proof}


\section{PKEET-FA over Ideal Lattices}\label{sec:construction-ideal}
The parameters of the scheme are $n,m,q,k$ and $\sigma, \alpha, \gamma, \tau$ and $\zeta$ are real numbers, and chosen as in Section~\ref{sec:Params-ideal}.
\subsubsection*{Construction}\label{sec:construction}
\begin{description}
	\item[$\Setup(1^n)$] $\\$
	On input the security parameter $1^n$, do the following:
	\begin{enumerate}
		\item Compute $\bf{a}\in R_q^m$ associated to its trapdoor $T_\bf{a}\in R^{(m-k)\times k}$, $(\bf{a},T_\bf{a})\gets\TrapGen(q,\sigma,h=0)$, i.e., 
		$\bf{a}=((\bf{a}')^{T} | -(\bf{a}')^{T}T_\bf{a})^T$
		\item Compute $\bf{b}\in R_q^m$ associated to its trapdoor $T_\bf{b}\in R^{(m-k)\times k}$, $(\bf{b},T_\bf{b})\gets\TrapGen(q,\sigma,h=0)$, i.e., 
		$\bf{b}=((\bf{b}')^{T} | -(\bf{b}')^{T}T_\bf{b})^T$
		\item Sample uniformly random $u\hookleftarrow U(R_q)$.
		\item Output $\PK = (\bf{a},\bf{b},u)\in R_q^{2m+1}$ and $\SK=(T_\bf{a},T_\bf{b})$
	\end{enumerate} 
	
	\item[$\Enc(\PK,M)$] $\\$
	 Given a message $M\in R_2$, do the following
	\begin{enumerate}
		\item Sample $s_1,s_2\hookleftarrow U(R_q)$, $e'_1,e'_2\hookleftarrow D_{R,\tau}$ and compute
		\begin{align*}
		\CT_1 &= u\cdot s_1+e'_1 +M\cdot\lfloor q/2\rfloor\in R_q,\\
	\CT_2&=u\cdot s_2 +e'_2 + H'(M)\cdot\lfloor q/2\rfloor\in R_q.
		\end{align*}
		
		where $H'$ is a hash function mapping from $\{ 0,1\}^*$ to the message space $\cal{M}$.
		\item Choose a random $v\in\ZZ_q^n$ and compute ${h}=H(v)\in R_q$.
		\item Compute $\bf{a}_h = \bf{a}^T + (\bf{0} | h\bf{g})^T = ((\bf{a}')^T | h\bf{g}-(\bf{a}')^TT_\bf{a})^T$.
		\item Compute $\bf{b}_h = \bf{b}^T + (\bf{0} | h\bf{g})^T = ((\bf{b}')^T | h\bf{g}-(\bf{b}')^TT_\bf{b})^T$.
		\item Choose $\bf{y},\bf{y}'\hookleftarrow D_{R^{m-k},\tau}, \bf{z}, \bf{z}'\in D_{R^k,\gamma}$ and compute
		\begin{align*}
		\CT_3 &= \bf{a}_h\cdot s_1 +(\bf{y}^T,\bf{z}^T)^T\in R_q^m,\\
		\CT_4 &= \bf{b}_h\cdot s_2 +((\bf{y}')^T,(\bf{z}')^T)^T\in R_q^m.
		\end{align*}
		Output the ciphertext
		$$\CT=(v,\CT_1,\CT_2,\CT_3,\CT_4)\in R_q^{2m+3}.$$
	\end{enumerate}
	\item[$\Dec(\SK,\CT)$] $\\$
	On input the secret key $\SK=(T_\bf{a},T_\bf{b})$ and a ciphertext $\CT=(v,\CT_1,\CT_2,\CT_3,\CT_4)$, do the following:
	\begin{enumerate}
		\item Compute $h=H(v)$ and construct $\bf{a}_h$ and $\bf{b}_h$ as in Step 3 and 4 in the Encryption process.
		\item Sample short vectors $\bf{x},\bf{x}'\in R_q^m$:
		\begin{align*}
		\bf{x} &\gets\SamplePre(T_\bf{a},\bf{a}_h,h,\zeta,\sigma,\alpha,u) \\
		\bf{x}' &\gets\SamplePre(T_\bf{b},\bf{b}_h,h,\zeta,\sigma,\alpha,u)
		\end{align*}
		\item Compute $\bf{w}=\CT_1-\CT_3^T\bf{x}\in R_q$.
		\item For each $w_i$, if it is closer to $\lfloor q/2\rfloor$ than to $0$, then output $M_i=1$, otherwise $M_i=0$. Then we obtain the message $M$.
		\item Compute $\bf{w}'=\CT_2-\CT_4^T\bf{x}'\in R_q$.
		\item For each $w'_i$, if it is closer to $\lfloor q/2\rfloor$ than to $0$, then output $\bf{h}_i=1$, otherwise $\bf{h}_i=0$. Then we obtain an element $\bf{h}$.
		\item If $\bf{h} = H'(M)$ then output $M$, otherwise output $\perp$.
	\end{enumerate}

Let $U_i$ and $U_j$ be two users of the system. We denote by $\CT_i = (\CT_{i,1},\CT_{i,2},\CT_{i,3},\CT_{i,4})$ (resp. $\CT_j = (\CT_{j,1},\CT_{j,2},\CT_{j,3},\CT_{j,4})$) be a ciphertext of $U_i$ (resp. $U_j$).\\

\item[Type-1 Authorization]~
\begin{itemize}
	\item $\Td_1(\SK_i)$: 
	On input a user $U_i$'s secret key $\SK_i = (T_{i,\bf{a}}, T_{i,\bf{b}})$, it outputs a trapdoor $\td_{1,i}=T_{i,\bf{b}}$.
	
	\item $\Test(\td_{1,i},\td_{1,j},\CT_i,\CT_j)$: 
	On input trapdoors $\td_{1,i},\td_{1,j}$ and ciphertexts $\CT_i,\CT_j$ for users $U_i, U_j$ respectively, computes
	\begin{enumerate}
		\item For each $i$ (resp. $j$), do the following
		\begin{enumerate}
			\item Compute $h_i=H(v_i)$ and sample $\bf{x'}_i\in R_q^m$ from
			$$\bf{x'}_i\gets\SamplePre(T_{i,\bf{b}},\bf{b}_h,h_i,\zeta,\sigma,\alpha,u).$$
			\item Compute $\bf{w}_i=\CT_{i,2}-\CT_{i,4}^T\bf{x'}_i$. For each $k=1,\cdots,n$, if $w_{i,k}$ is closer to  $\lfloor q/2\rfloor$ than to $0$, then output $\bf{h}_{ik}=1$, otherwise $\bf{h}_{ik}=0$. Then we obtain the element $\bf{h}_i$ (resp. $\bf{h}_j$).
			
		\end{enumerate}
		\item Output $1$ if $\bf{h}_i=\bf{h}_j$, and $0$ otherwise.
	\end{enumerate}
\end{itemize}

\item[Type-2 Authorization]~
\begin{itemize}
	\item $\Td_2(\SK_i,\CT_i)$: 
	On input a user $U_i$'s secret key $\SK_i = (T_{i,\bf{a}}, T_{i,\bf{b}})$ and ciphertext $\CT_i=(v_i,\CT_{i,1},\CT_{i,2},\CT_{i,3},\CT_{i,4})$, it samples $\bf{x'}_i\in R_q^m$ from
	$$\bf{x'}_i\gets\SamplePre(T_{i,\bf{b}},\bf{b}_h,h_i,\zeta,\sigma,\alpha,u).$$
	it outputs a trapdoor $\td_{2,i}=\bf{x'}_i$.
	
	\item $\Test(\td_{2,i},\td_{2,j},\CT_i,\CT_j)$: 
	On input trapdoors $\td_{2,i},\td_{2,j}$ and ciphertexts $\CT_i,\CT_j$ for users $U_i, U_j$ respectively, do the following
	\begin{enumerate}
		 \item Compute $\bf{w}_i=\CT_{i,2}-\CT_{i,4}^T\bf{x'}_i$. For each $k=1,\cdots,n$, if $w_{i,k}$ is closer to  $\lfloor q/2\rfloor$ than to $0$, then output $\bf{h}_{ik}=1$, otherwise $\bf{h}_{ik}=0$. Then we obtain the element $\bf{h}_i$ (resp. $\bf{h}_j$).
		\item Output $1$ if $\bf{h}_i=\bf{h}_j$, and $0$ otherwise.
	\end{enumerate}
\end{itemize}

\item[Type-3 Authorization]~
\begin{itemize}
	\item $\Td_{3,i}(\SK_i,\CT_i)$: 
	On input a user $U_i$'s secret key $\SK_i = (T_{i,\bf{a}}, T_{i,\bf{b}})$ and ciphertext $\CT_i$, it outputs a trapdoor $\td_{3,i}=\bf{x'}_i$ by sampling 
	$$\bf{x'}_i\gets\SamplePre(T_{i,\bf{b}},\bf{b}_h,h_i,\zeta,\sigma,\alpha,u).$$
	\item $\Td_{3,j}(\SK_j)$: On input a user $U_j$'s secret key $\SK_j = (T_{j,\bf{a}}, T_{j,\bf{b}})$, it outputs a trapdoor $\td_{3,j}=T_{j,\bf{b}}$.
	\item $\Test(\td_{i,3}\td_{j,3},\CT_i,\CT_j)$: 
	On input trapdoors $\td_{3,i},\td_{j,3}$ and ciphertexts $\CT_i,\CT_j$ for users $U_i, U_j$ respectively, do the following
	\begin{enumerate}
		\item Compute $\bf{w}_i=\CT_{i,2}-\CT_{i,4}^T\bf{x'}_i$. For each $k=1,\cdots,n$, if $w_{i,k}$ is closer to  $\lfloor q/2\rfloor$ than to $0$, then output $\bf{h}_{ik}=1$, otherwise $\bf{h}_{ik}=0$. Then we obtain the element $\bf{h}_i$.
			\item Compute $h_j=H(v_j)$ and sample $\bf{x'}_j\in R_q^m$ from
			$$\bf{x'}_j\gets\SamplePre(T_{j,\bf{b}},\bf{b}_h,h_j,\zeta,\sigma,\alpha,u).$$
			\item Compute $\bf{w}_j=\CT_{j,2}-\CT_{j,4}^T\bf{x'}_j$. For each $k=1,\cdots,n$, if $w_{j,k}$ is closer to  $\lfloor q/2\rfloor$ than to $0$, then output $\bf{h}_{jk}=1$, otherwise $\bf{h}_{jk}=0$. Then we obtain the element $\bf{h}_j$.

		\item Output $1$ if $\bf{h}_i=\bf{h}_j$, and $0$ otherwise.
	\end{enumerate}
\end{itemize}
\end{description}

\begin{lem}[Correctness]
	With the choice of parameters as in~\ref{sec:Params-ideal}, our proposed PKEET  is correct,  assuming that the hash function $H'$ is collision-resitant.
\end{lem}

\begin{proof}
	Let $\bf{x}=(\bf{x}_0^T|\bf{x}_1^T)^T$ with $\bf{x}_0\in R_q^{m-k}$ and $\bf{x}_1\in R_q^k$. To correctly decrypt a ciphertext, we need the error term $e_1'-(\bf{y}^T|\bf{z}^T)(\bf{x}_0^T|\bf{x}_1^T)^T = e_1'-\bf{y}^T\bf{x}_0-\bf{z}^T\bf{x}_1$ to be bounded by $\lfloor q/4\rfloor$, which is satisfied by the choice of parameters in~\ref{sec:Params-ideal}. Similarly, for the test procedure, one needs to correctly decrypt $H'(M)$ and the equality test works correctly given that $H'$ is collision-resistant.
\end{proof}

\subsubsection*{Security analysis}
In this section, we will prove that our proposed scheme is $\OWa$ secure against Type-I adversaries (cf.~Theorem~\ref{thm:OWCPA}) and $\INDa$ secure against Type-II adversaries (cf. Theorem~\ref{thm:INDCPA}).

\begin{thm}[$\OWa$]\label{thm:OWCPA}
	The proposed PKEET scheme with parameter $(q,n,m,\sigma,\alpha)$ as in Section~\ref{sec:Params-ideal} is $\OWa$ secure provided that $H'$ is one-way hash function, $H$ is a collision resistant hash function  and the $\RLWE$ problem is hard. In particular, suppose there exists a probabilistic algorithm $\cal{A}$ that wins the $\OWa$ game with advantage $\epsilon$ then there is a probabilistic algorithm $\cal{B}$ that solves that $\RLWE$ problem with advantage $\epsilon'$ such that
	$$\epsilon'\geq \epsilon -\epsilon_{H',\mathsf{OW}}-\epsilon_{H,\mathsf{CR}}$$
	where $\epsilon_{H',\mathsf{OW}}$ and $\epsilon_{H,\mathsf{CR}}$ are the advantage of breaking the one-wayness of $H'$ and the collision resistance of $H$ respectively.
\end{thm}
\begin{proof}
The proof follows that of {\cite[Theorem 5]{Duong19}} and {\cite[Theorem 1]{BertFRS18-implement}}.
	Assume that there exists a Type-I adversary $\cal{A}$ who breaks the $\OWa$ security of the PKEET scheme with non-negligible probability $\epsilon$. We construct an algorithm $\cal{B}$ who solves the RLWE problem using $\cal{A}$. Assume again that there are $N$ users in our PKEET system. We now describe the behaviors of $\cal{B}$. Assume that $\theta$ is the target index of the adversary $\cal{A}$ and the challenge ciphertext is $\CT^*_\theta=(v^*,\CT^*_{\theta,1},\CT^*_{\theta,2},\CT^*_{\theta,3},\CT^*_{\theta,4})$.
	
	We will proceed the proof in a sequence of games. In Game $i$, let $W_i$ denote the event that the adversary $\cal{A}$ win the game. The adversary's advantage in Game $i$ is $\Pr[W_i]$.
	
\begin{description}
	\item[\textbf{Game 0}.]~~ This is the original $\OWa$ game between the attacker $\cal{A}$ against the scheme and the $\OWa$ challenger.
	
	\item[\textbf{Game 1}.]~~ This is similar to Game $0$ except that at
	the challenge phase, $\cal{B}$ chooses two message $M$ and $M'$ in the message space and encrypt $M$ in $\CT^*_{\theta,1}$ and $H'(M')$ in $\CT^*_{\theta,2}$. Other steps are similar to Game 0. Since $\cal{A}$ may have a trapdoor $\Td_{\a}$ (for $\a = 1, 2$) then he can obtain $H'(M')$. At the end, if $\cal{A}$ outputs $M'$, call this event $E_1$, then $\cal{A}$ has broken the one-wayness of $H'$. Thus $\Pr[E_1]\leq\epsilon_{H',\mathsf{OW}}$ where $\epsilon_{H',\mathsf{OW}}$ is the advantage of $\cal{A}$ in breaking the one-wayness of $H'$. Therefore, one has 
	$$\Pr[W_0]-\Pr[W_1]\leq \epsilon_{H',\mathsf{OW}}.$$
	
	\item[\textbf{Game 2}.]~~ This is similar to Game 1 except the way the challenge $\cal{B}$ generates the public key for the user with index $\theta$, and the challenge ciphertext $\CT^*_\t$ as the following. At the start of the experiment, choose a random $v^*\in\ZZ_q^n$ and let  the public parameter $\bf{a}$ generated by $\TrapGen(q,\sigma,\bf{a}',-h^*_{\theta})$ where $h^*_{\theta}=H(v^*)$. Hence, the public parameter is $\bf{a}=((\bf{a}')^T | -h^*_{\theta}\bf{g}-(\bf{a}')^TT_\bf{a})^T$, where the first part $\bf{a}'\in R_q^{m-k}$ is chosen from the uniform distribution. For the second part $\bf{a}'^TT_\bf{a}=(\sum_{i=1}^{m-k}a_it_{i,1},\cdots,\sum_{i=1}^{m-k}a_it_{i,k})$ is indistinguishable from the uniform distribution.
	In our paper, we choose $m-k=2$ and $\bf{a}'=(1,a)$ with $a\hookleftarrow U(R_q)$ and the public key $\bf{a}=(1,a|-(at_{2,1}+t_{1,1}),\cdots,-(at_{2,k}+t_{1,k}))$ looks uniform followed by the RLWE assumption, given that the secret and error follow the same distribution. The remainder of the game is unchanged and similar to Game 1. 
	
	Note that whenever $\cal{A}$ queries $\cal{O}^{\Dec}(\t,\CT_\t)$ with $\CT_\t=(v,\CT_{\t,1},\CT_{\t,2},\CT_{\t,3},\CT_{\t,4})$ then $\cal{B}$ does as follows. If $v=0$ or $v=v^*$ then $\cal{B}$ aborts. Otherwise, $\cal{B}$ can answer as usual using the trapdoor $T_\bf{a}$, except if $H(v)=h^*_\t$, which happens with probability at most the advantage $\epsilon_{H,\mathsf{CR}}$ of breaking the collision-resistance of $H$.
	It follows that $$\Pr[W_2]-\Pr[W_1]\leq \epsilon_{H,\mathsf{CR}}.$$
	
	\item[\textbf{Game 3}.]~~ In this game, the challenge ciphertext $\CT^*_\theta$ is now chosen uniformly in $R_q^{2m+3}$. We now show that Game 3 and Game 2 are indistinguishable for $\bf{A}$ by doing a reduction from RLWE problem. 
	\end{description}
	
	Now $\cal{B}$ receives $m-k+1$ samples $(a_i,b_i)_{0\leq i\leq m-k}$ as an instance of the decisional RLWE problem. Let $\bf{a}'=(a_1,\cdots,a_{m-k})^T\in R_q^{m-k}$ and $\bf{b}'=(b_1,\cdots,b_{m-k})^T\in R_q^{m-k}$. The simulator runs $\TrapGen(q,\sigma,\bf{a}',-h^*_\theta)$, and we get $\bf{a}=((\bf{a}')^T | -h^*_{\theta}\bf{g}-(\bf{a}')^TT_{\bf{a}})^T$ as in Game 2. Similarly, the simulator runs $\TrapGen(q,\sigma,\bf{b}',-h^*_\theta)$ to get $\bf{b}=((\bf{b}')^T | -h^*_{\theta}\bf{g}-(\bf{b}')^TT_{\bf{a}})^T$.  Next $\cal{B}$ set $u=a_0$ and sends $\PK_\theta=(\bf{a},\bf{b},u)$ to $\cal{A}$ as the public key of the user $\theta$. 
	
	At the challenge phase, the simulator chooses a message $M$ and computes the challenge ciphertext $\CT^*_\theta\gets\Enc(\PK_\theta,M)$ as follows:
	\begin{enumerate}
		\item Set $\CT^*_{\theta,1}\gets b_0 + M\cdot\lfloor q/2\rfloor.$
	\item Choose a uniformly random $s_2\in R_q$ and $e_2'\hookleftarrow D_{R,\tau}$ and compute
	$$\CT^*_{\theta,2} = u\cdot s_2 + e_2' + H'(M)\cdot\lfloor q/2\rfloor\in R_q.$$

		\item Set 
		$$\CT^*_{\t,3}=\left[ 
		\begin{array}{c}
		\bf{b}' \\
		-\bf{b}'T_\bf{a} +\widehat{\bf{e}}
		\end{array}
		 \right]\in R_q^m$$
	with $\widehat{\bf{e}}\hookleftarrow D_{R_q^k,\,u}$ for some real $\mu$.
		\item Compute $h^*_\t = H(v^*)\in R_q$.
	\item Choose $\bf{y}'\hookleftarrow D_{R^{m-k},\tau}$, $\bf{z}'\hookleftarrow R_{R^k,\gamma}$ and set
	$$\CT^*_{\t,4}=\bf{b}_{h_\t}\cdot s_2+((\bf{y}')^T,(\bf{z}')^T)^T\in R_q^m.$$
	\end{enumerate}		
Then $\cal{B}$ sends $\CT^*_\t=(v^*,\CT^*_{\t,1},\CT^*_{\t,2},\CT^*_{\t,3},\CT^*_{\t,4})$ to $\cal{A}$.

When the samples $(a_i,b_i)$ are LWE samples, then $\bf{b}'=\bf{a}'s_1+\bf{e}'$ and $b_0=a_0s_1+e_0$ for some $s_1\in R_q$ and $e_0\hookleftarrow D_{R,\tau}$, $\bf{e}'\hookleftarrow D_{R^{m-k},\tau}$. It implies that 
$$
\CT^*_{\t,1} = u\cdot s_1+e_0 + M\cdot\lfloor q/2\rfloor$$
and
$$\CT^*_{\t,3} = \bf{a}_{h_\t}\cdot s_1 + (\bf{e}'^T | \bf{z}^T) $$
where $\bf{z}=-\bf{e}'^TT_{\bf{a}}+\widehat{\bf{e}}^T$ is indistinguishable from a sample drawn from the distribution $D_{R^k,\gamma}$ with $\gamma^2=(\sigma\|\bf{e}'\|)^2+\mu^2$ for $\mu$ well chosen. 

Then $\CT^*_{\t}$ is a valid ciphertext.

When the $(a_i,b_i)$ are uniformly random in $R_q^2$, then obviously $\CT^*_{\t}$ also looks uniform.

$\cal{A}$ guesses if it is interacting with Game 3 or Game 2. The simulator outputs the final guess as the answer to the
RLWE problem. One can easily obtain that 
$$\Pr[W_3]-\Pr[W_2]\leq \epsilon'.$$
Combining the above results we obtain 
$$\epsilon = \Pr[W_0] \leq \epsilon_{H',\mathsf{OW}}+\epsilon_{H,\mathsf{CR}}+\epsilon'$$
which implies
$$\epsilon'\geq \epsilon-\epsilon_{H,\mathsf{CR}}-\epsilon_{H',\mathsf{OW}} .$$
\end{proof}

\begin{thm}[$\INDa$]\label{thm:INDCPA}
	The proposed PKEET scheme with parameter $(q,n,m,\sigma,\alpha)$ as in... is $\INDa$ secure provided that $H$ is a one-way hash function, $H'$ is a collision resistant hash function and the $\RLWE$ is hard. In particular, suppose there exists a probabilistic algorithm $\cal{A}$ that wins the $\INDa$ game with advantage $\epsilon$ then there is a probabilistic algorithm $\cal{B}$ that solves that $\RLWE$ problem with advantage $\epsilon'$ such that
	$$\epsilon'\geq \epsilon -\epsilon_{H,\mathsf{CR}},$$
	where $\epsilon_{H,\mathsf{CR}}$ is the advantage of breaking the collision resistance of $H$.
\end{thm}
\begin{proof}
	The proof is similar to that of Theorem~\ref{thm:OWCPA} and that of~{\cite[Theorem 6]{Duong19}}, hence we omit the proof here. Note that in this proof, we do not consider Game 1 as in Theorem~\ref{thm:OWCPA} , which results in not having $\epsilon_{H',\mathsf{OW}}$ in the advantage formula.
\end{proof}
\subsubsection*{Parameters}\label{sec:Params-ideal}
We follow~{\cite[Section 4.2]{BertFRS18-implement}} for choosing parameters for our scheme as the following:
\begin{enumerate}
	\item The modulus $q$ is choosen to be a prime of size $62$ bits.
	\item We choose $m-k=2$.
	\item The Gaussian parameter $\sigma$ for the trapdoor sampling is $\sigma>\sqrt{(\ln(2n/\epsilon)/\pi)}$ (\cite{MP12}) where $n$ is the maximum length of the ring polynomials, and $\epsilon$ is the desired bound on the statistical error introduced by each randomized rounding operation. This parameter is also chosen to ensure the hardness of $\RLWE$ problem.
	\item The Gaussian parameter $\sigma$ for the $G$-sampling is $\alpha = \sqrt{5}\sigma$ (\cite{MP12}).
	\item The parameter $\zeta$ is chosen such that $\zeta>\sqrt{5}C\sigma^2(\sqrt{kn}+\sqrt{2n}+t')$ for $C\cong 1/\sqrt{2\pi}$ and $t'\geq 0$, following~\cite{GM18}.
	\item For decrypting correctly, we need
$$
	t\tau\sqrt{n}+2t^2\tau\zeta n+t^2\gamma\zeta k n<\lfloor q/4\rfloor.
$$
\item Finally, we choose $\mu=t\sigma\tau\sqrt{2n}$ and $\gamma=2t\sigma\tau\sqrt{n}$ so that $\gamma$ satisfies $\gamma^2=(\sigma\|\bf{e}'\|)^2+\mu^2$.

\item The parameter $t$ here is chosen such that a vector $\bf{x}$ sampled in $D_{\ZZ^m,\sigma}$ has norm $\|\bf{x}\|\leq t\sigma\sqrt{m}$. Note that
$$\Pr_{x\hookleftarrow D_{\ZZ,\sigma}}[|x|>t\sigma]\leq\mathrm{erfc}(t/\sqrt{2})$$
with $\mathrm{erfc}(x)=1-\frac{2}{\pi}\int_{0}^x\exp^{-t^2}dt$. One can choose, for example, $t=12$ (see {\cite[Section 2]{BertFRS18-implement}}).
\end{enumerate}

\section{Discussion}
PKEET-FA over integer lattices has lowest ciphertext and secret key sizes. But, public key sizes is lowest in case of PKEET-FA over ideal lattices. We have provided a comparative study of data sizes among the proposed schemes in table \ref{tab2}.

	{\renewcommand{\arraystretch}{1.2}
		\setlength{\tabcolsep}{5pt}
		\begin{table*}[ht]
			\footnotesize
			\begin{center}
				\caption{Comparison among Proposed PKEET-FA.}
				\begin{tabular}{c c c c }
					\toprule
					Scheme & Ciphertext & Public Key & Secret Key  \\ 
					\midrule
					
					Section \ref{sec:PKEET-FA} & $l + (2t + 4m)\log q$ &  $((l+3)mn+nt)\log q$ & $2m^2\log q$ \\
					Section \ref{sec:construct Lee} &  $(8m+2t +2mt) \log q$ &  $(4mn+nt)\log q$ & $2m^2\log q$ \\
					Section \ref{sec:construction-ideal} & $n(2m+3)\log q$ & $n(2m+1)\log q$  & $2nk(m-k)\log q$ \\
					\toprule
					
					\multicolumn{4}{l}{ \parbox[t]{0.5\textwidth}{Data sizes are in number of bits. }} \\
				\end{tabular}
				\label{tab2}
			\end{center}
			
		\end{table*}
	}
	\normalsize

\section{Implementation of PKEET-FA over Ideal Lattices}
We discuss here a small test implementation we created. The following code hosted at \url{https:\\github.com\TBD}
The purpose of this implementation is 
to serve as a baseline for further efficiency improvements. 

\begin{table}[!h]
	\centering
\begin{tabular}{|c|| c | c | c | c | c | c | c | c | c |}
	\hline
	Test & Setup & Encrypt & Decrypt & $Td_1$ & $Td_2$ & $Td_{3,(i,j)}$ & $Test_1$ & $Test_2$ & $Test_3$\\
	\hline
	Time (ms) & 4.644 & 7.772 & 38.618 & 0,0001 & 37.715 & 18.647 & 37.776 & 0.8203 & 19.648\\
	\hline
    \end{tabular}
    
    \vspace*{0.1in}
	\caption{Test results, average time in ms after 1000 runs, security $\lambda = 195$}
	\label{table::test_results}
\end{table}

While the whole program architecture is different, we attempted to keep all the primitives identical from sampling to memory allocations.
The results of the computations are shown in the table \ref{table::test_results}.
The computations were done on a \textit{Intel(R) Core(TM) i7-8665U CPU @ 1.90GHz} processor using the \textit{Windows Subsystem for Linux}.
Note that our timing tests were done sequentially, without the use of multithreading to run test samples as it was the case in \url{https://github.com/lbibe/code/blob/master/src/main.cpp}. Some primitives use multithreading whenever it was also used on the available code of lbibe. Note that the thread number was set to 2. The results presented in table \ref{table::test_results} are consistent according to the design of the scheme. 

We do not provide any comparison in this paper: most comparisons available online and used in the current literature are NIST candidates, which follow specific requirements and have various degrees of optimizations that are, in our honest opinion, not consistent between schemes. 
As far as we know, no NIST submissions have an equality test implemented.


\bibliographystyle{splncs04}
\bibliography{latbib}

	\end{document}